\documentclass[twoside,11pt]{article}

\newif\ifarxiv
\arxivtrue

\ifarxiv
  \usepackage[preprint,abbrvbib]{jmlr2e}
\else
  \usepackage[abbrvbib]{jmlr2e}
\fi

\usepackage{amsmath}
\usepackage{graphicx}
\usepackage{booktabs}
\usepackage{multirow}
\usepackage{array}
\usepackage{placeins}
\newcolumntype{L}[1]{>{\raggedright\arraybackslash}p{#1}}

\newtheorem{assumption}{Assumption}

\newcommand{\E}{\mathbb{E}}
\newcommand{\Pp}{\mathbb{P}}
\newcommand{\Pq}{\mathbb{P}_{Q}}
\newcommand{\R}{\mathbb{R}}
\newcommand{\Cov}{\operatorname{Cov}}
\newcommand{\Corr}{\operatorname{Corr}}
\newcommand{\Var}{\operatorname{Var}}
\newcommand{\sign}{\operatorname{sign}}
\newcommand{\tr}{\operatorname{tr}}
\newcommand{\mnr}{\mathrm{MNR}}

\ifarxiv
  \jmlrheading{}{2026}{1-\pageref{LastPage}}{}{}{}{Wenxuan Xiao and Xu Cao}
\else
  \jmlrheading{1}{2027}{1-\pageref{LastPage}}{}{}{26-XXXX}{Wenxuan Xiao and Xu Cao}
\fi
\ShortHeadings{Low-Bit Decision Stability in Vector Search}{Xiao and Cao}
\firstpageno{1}
\usepackage{lastpage}

\begin{document}

\title{When Does Low-Bit Quantization Preserve the Decisions of Vector Search?}

\author{\name Wenxuan Xiao\thanks{Corresponding author.} \email w.xiao@astrmira.com \\
       \addr Astrmira Tech.\\
       Liuyanghe Subdistrict, Changsha, Hunan, China
       \AND
       \name Xu Cao \email x.cao@astrmira.com \\
       \addr Astrmira Tech.}

\ifarxiv
  \editor{}
\else
  \editor{TBD}
\fi

\maketitle

\begin{abstract}%
An orthogonal rotation improves the distance rank correlation of two-bit Cohere codes from 0.75 to 0.91, yet increases nearest-neighbour comparison flips by 51\%. We analyse quantization through the comparisons made by graph search. A comparison changes only when its quantization error exceeds its exact margin, giving a flip bound that separates close calls from large errors. Shared endpoints induce correlated distance errors that partly cancel within a comparison. An exact replay identity connects candidate-level actions to the output of neighbour selection. A Gaussian model gives coordinate-code errors in closed form, and a rare-spike construction exposes heavy tails in the population-radius error. Held-out blocks give model-free certificates for acceptance rules based on exact or quantized margins. Across 24 dataset--code configurations, calibration-estimated risk correlates with held-out ranking and pruning flips at 0.97 and 0.99, compared with 0.91 and 0.99 for a margin-to-noise summary and 0.74 and 0.05 for distance rank correlation. Paired builds of 200{,}000-vector graphs connect rotation, query scoring, and topology to end-to-end recall and search work.
\end{abstract}


% ============================================================
\section{Introduction}
\label{sec:intro}

Low-bit codes are the main tool for making graph-based vector search fast and compact. Binary quantization stores a 768-dimensional vector in 96 bytes and replaces most floating-point distance arithmetic with bit operations. How well such a code serves search, however, is hard to read from how well it estimates distances. Consider two-bit codes of Cohere text embeddings, with and without an orthogonal rotation applied first. On the same 160 held-out queries and candidate sets, rotation raises the distance rank correlation from 0.747 to 0.907. Yet when each query's exact nearest neighbour is compared with its other candidates, the quantized comparison gives the wrong answer 8.81\% of the time instead of 5.83\%, about half again as often. On other decision sets rotation leaves ranking flips nearly unchanged and lowers pruning flips. We call this the \emph{rotation paradox}. Distance fidelity alone does not determine the direction of change in decision errors.

The effect of a code also depends strongly on the data. Recent systems use low-bit codes in different ways: QuIVer keeps the coordinate axes of the encoder and builds and searches a proximity graph on two-bit scores~\citep{xiao2026quiver}, while RaBitQ rotates vectors randomly before one-bit coding and corrects its distance estimates~\citep{gao2023rabitq}. At the same search width, QuIVer finds 99.19\% of the ten nearest neighbours on Cohere text embeddings and 3.28\% on GIST image descriptors~\citep{xiao2026quiver}. A practitioner choosing a code for a new embedding needs to anticipate differences of this kind without building and searching an index for every candidate code. The usual summaries of quantization quality measure \emph{fidelity}, how closely approximate distances track exact ones, through mean squared error or rank correlation, and they predict search failures poorly: across the 24 dataset--code configurations of Section~\ref{sec:experiments}, the rank correlation between approximate and exact distances predicts how often held-out pruning decisions change with a Spearman correlation of only 0.05.

This paper explains such outcomes by analysing quantization at the level at which a graph algorithm uses it. A graph search never uses a distance on its own; it compares distances. Is candidate $x$ closer to the query than candidate $y$? Is a candidate closer to an already selected neighbour than to the node being linked? Each comparison has an exact margin, the gap between the two exact quantities, and the quantized comparison gives the same answer unless its error exceeds that margin. Two consequences follow.

First, the comparisons that matter are close calls. A graph search compares candidates that all lie near the query, so small margins are common. On Cohere two-bit codes, comparisons whose margin is below a quarter of the error scale change 48.6\% of the time, while none of the 309 comparisons with a margin above five times the error scale changed.

Second, the relevant error is the error of a difference, and the two compared distances usually share an endpoint. Quantizing a query can shift all of its distances in a similar way, and that shared part cancels within a comparison. On Cohere two-bit codes, comparison errors have 2.4 to 7.4 times less variance than independent distance errors would, depending on the calibration. On MiniLM the factor is 1.05 to 1.10. Rotation changes the covariance, bias, and scale of the residual jointly. Section~\ref{sec:exp-rotation} crosses two calibrations and shows why error variance must be considered together with decision margins.

We organize the analysis around three questions.
\begin{enumerate}
\setlength{\itemsep}{2pt}
\setlength{\parsep}{0pt}
\item \textbf{When does a single comparison change?} (Section~\ref{sec:flip}) The probability of a flip is at most the probability of a close call plus the probability of a large comparison error, with no distributional assumption (Theorem~\ref{thm:main}). We work out the comparison error for ranking and for Vamana pruning, where shared endpoints make errors correlated and calibration offsets add a bias.
\item \textbf{When does a whole neighbour list change?} (Section~\ref{sec:trace}) A Vamana neighbour-selection call returns a different list exactly when its decision on some candidate, replayed on the exact run, changes, and the first such candidate is where the two lists start to differ (Theorem~\ref{thm:trace}). The replay identifies the comparisons whose risks enter edge, list, and path bounds and diagnostic scores.
\item \textbf{How large are comparison errors, and how can we know?} (Sections~\ref{sec:bridge} and~\ref{sec:heldout}) For embeddings whose coordinates behave like Gaussian vectors, Stein's lemma gives the error structure of sign and two-bit codes in closed form, and the same model shows when a second, magnitude bit improves the code (Proposition~\ref{prop:magnitude}). A rare-spike construction identifies a tail obstruction for the population-radius idealization (Theorem~\ref{thm:counterexample}). When no model applies, independent held-out data bound the failure rate of the accepted decisions directly (Theorem~\ref{thm:heldout}).
\end{enumerate}

The experiments test each answer on learned, classical, and synthetic embeddings. Across 24 dataset--code configurations, a calibration-only close-call risk score correlates with held-out flip rates with Spearman correlation 0.97 for ranking and 0.99 for pruning, against 0.74 and 0.05 for distance rank correlation. A simple margin-to-noise summary gives 0.91 and 0.99, with paired confidence intervals that overlap the full score. An implementation check reproduces the first-divergence identity in all 960 sampled calls, and held-out certificates cover rules based on quantized margins. Paired 200{,}000-vector graph builds then measure how rotation and query scoring affect recall and work. Section~\ref{sec:instantiations} places RaBitQ, Lucene BBQ, and product quantization in the same framework.

\noindent\emph{Contributions.}\quad
\begin{itemize}
\setlength{\itemsep}{2pt}
\setlength{\parsep}{0pt}
\item A decision-level account of quantization error in graph search: a close-call plus large-error bound for single comparisons, error formulas for ranking and pruning comparisons that share endpoints, and an exact first-divergence theorem for neighbour selection.
\item Two ways to handle the large-error term: a Gaussian model that gives the error of coordinate codes in closed form, with a construction that marks where it stops, and model-free held-out certificates.
\item Experiments on learned, classical, and synthetic embeddings and six quantizer families showing that decision-level quantities predict failures that distance-level fidelity misses, a resolution of the rotation paradox, and practical guidance for choosing a code.
\end{itemize}

The formal results concern comparisons and neighbour-selection calls over given candidate sets. Section~\ref{sec:exp-graph} complements them with actual graph construction and search, where the graph itself generates the candidates.

% ============================================================
\section{Setting}
\label{sec:setting}

Ranking and pruning compare distances through signed scores whose signs determine the selected actions.

\subsection{Decisions in Graph Search}
\label{sec:decisions}

Let $d(\cdot,\cdot)$ be an exact dissimilarity, the cosine distance for normalized embeddings, and let $\hat d_Q$ be its approximation under a code $Q$. Graph search and graph construction make three kinds of decisions.
\begin{itemize}
\setlength{\itemsep}{2pt}
\setlength{\parsep}{0pt}
\item \textbf{Ranking.} Given a query $q$ and candidates $x$ and $y$, decide which is closer to $q$. The decision score is $\Psi_{\mathrm{rank}}=d(q,y)-d(q,x)$.
\item \textbf{Pruning.} When Vamana links a target $t$, it drops a candidate $c$ if some already selected neighbour $s$ is at least $\alpha\ge1$ times closer to $c$ than $t$ is~\citep{subramanya2019diskann}. The decision score is $\Psi_\alpha=d(t,c)-\alpha\,d(s,c)$, and $c$ is pruned when $\Psi_\alpha\ge0$.
\item \textbf{Stopping.} A search stops expanding when the best unexpanded candidate is no closer than the worst retained one; the decision score is the difference of these two distances.
\end{itemize}
In each case the sign of the score selects the action, with a fixed rule for a score of zero. A \emph{flip} occurs when the quantized score $\widehat\Psi_Q$ selects a different action from the exact score; nearest-vs-rest score audits also count a quantized tie against a nonzero exact score as a flip, while action replays and block certificates use their stated deterministic tie-breakers. The \emph{exact margin} $\Gamma=|\Psi|$ measures how close the call is.

\subsection{Comparison Error and Calibration}
\label{sec:calibration}

Quantized scores live on a different scale from exact ones. A two-bit score, for example, is an integer. For a positively scaled quantized comparison $\widehat\Psi_Q$, define the \emph{comparison error}
\begin{equation}
R=\widehat\Psi_Q-c_Q\Psi,
\label{eq:comparison-error}
\end{equation}
where $c_Q>0$ is the exact-signal scale in the units of $\widehat\Psi_Q$. At a nonzero exact margin, a flip requires $R$ to oppose $\Psi$ with magnitude at least $c_Q\Gamma$. In the experiments, an independently fitted map $a\hat d_{\rm raw}+b$ puts quantized distances in exact-distance units, so $c_Q=1$. Equivalently, dividing the calibrated score by $a$ gives $c_Q=1/a$, $R_{\rm raw}=R_{\rm metric}/a$, and $v_{\rm raw}/c_Q=v_{\rm metric}$. A \emph{distance-fitted} affine map uses individual distance pairs, sampled randomly or from fitting candidate pools. A \emph{comparison-fitted} positive gain acts on signed comparison scores. Appendix~\ref{app:calibration-direction} specifies each experiment's fitting population. An intercept cancels in ranking but contributes $(1-\alpha)b$ in pruning, so the mean of $R$ is carried explicitly in Section~\ref{sec:flip}. Positive rescaling preserves ranking flips, including quantized ties.

The \emph{error scale} $v$ is the standard deviation of $R$ on the calibration data, and the \emph{margin-to-noise ratio} of a comparison is
\begin{equation}
\mnr=\frac{c_Q\Gamma}{v}.
\end{equation}
A comparison with a small margin-to-noise ratio is a close call (Figure~\ref{fig:comparison-geometry}).

\subsection{What Is Random}
\label{sec:randomness}

The results of this paper are probability statements, and they draw randomness from two sources. The first is the \emph{decision population}: the comparisons a search makes on a workload, with the query drawn from the query distribution, the candidates chosen by the search, and, for pruning, a target with its candidate list. The second is the code itself, when the code is randomized.

For a deterministic code, the decision population supplies the randomness. A randomized code also varies its own seed. Propositions~\ref{prop:ranking}--\ref{prop:pruning} condition on vectors and randomize the code, while the Gaussian model samples embeddings and held-out certificates sample independent blocks. The close-call bound allows either probability space, and first divergence is deterministic.

Which comparisons enter the population matters. Comparisons selected by a search have small margins, because their candidates lie near the query, and their errors are more strongly correlated than those of random pairs (Section~\ref{sec:exp-covariance}).

\begin{figure}[tbp]
\centering
\includegraphics[width=\textwidth]{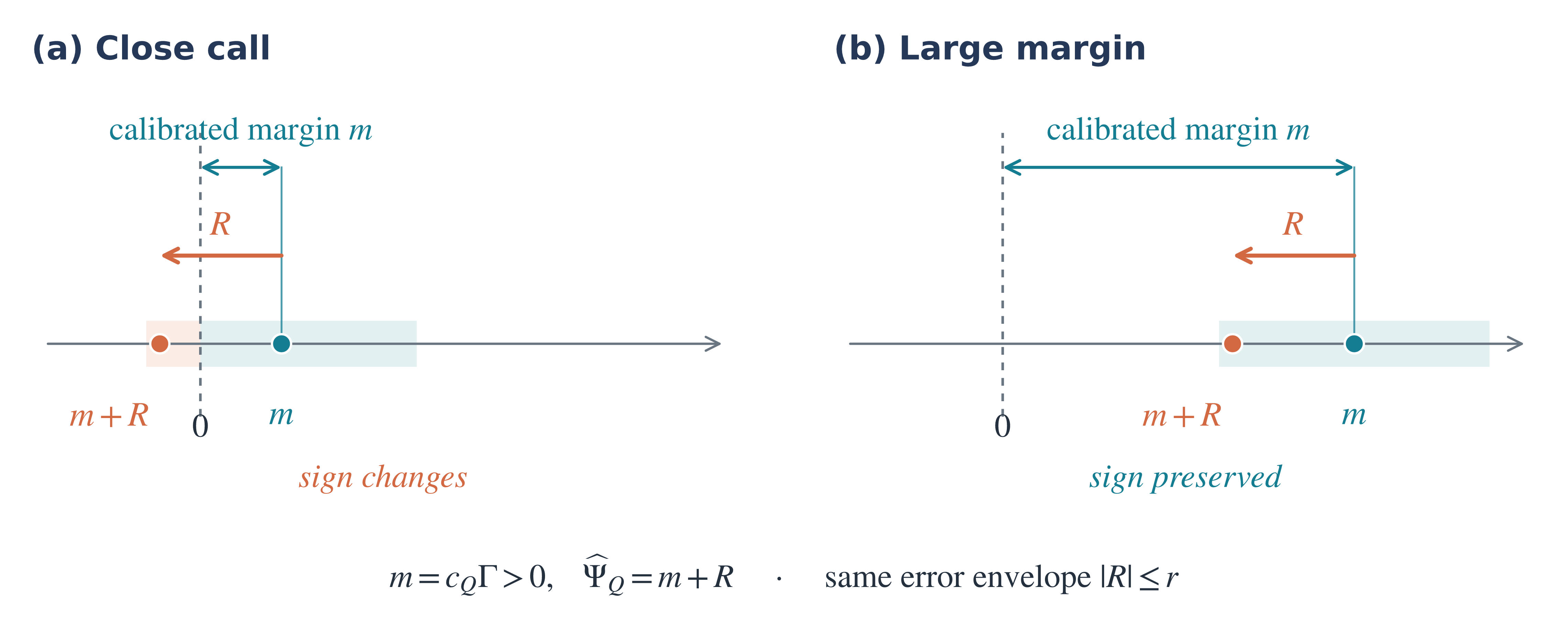}
\small
\caption{Equal errors at different margins. With $\Psi>0$, $m=c_Q\Gamma$ and $\widehat\Psi_Q=m+R$. The same illustrated error envelope $|R|\le r$ crosses zero for the close call but stays positive at the larger margin.}
\label{fig:comparison-geometry}
\end{figure}

% ============================================================
\section{When Does a Comparison Flip?}
\label{sec:flip}

A comparison flips only when its error is large relative to its margin, so its flip probability splits into how often margins are small, which depends only on the data and the decisions, and how often errors are large, which depends on the code applied to those decisions.

\subsection{Close Calls and Large Errors}
\label{sec:close-call}

Theorem~\ref{thm:main} makes this split exact. It holds on any probability space, and it needs no assumption about how errors are distributed or how they depend on the margin.

\begin{theorem}[Close-call bound]
\label{thm:main}
Put $R_{\mathrm{opp}}=-\sign(\Psi)R$, with $\sign(0)=0$. For every $\tau>0$,
\begin{equation}
\begin{aligned}
\Pp(\text{flip})
&\le \Pp(\Gamma\le\tau)+\Pp(R_{\mathrm{opp}}\ge c_Q\tau)\\
&\le \underbrace{\Pp(\Gamma\le\tau)}_{\text{close-call probability}}
+\underbrace{\Pp(|R|\ge c_Q\tau)}_{\text{large-error probability}}.
\end{aligned}
\label{eq:decomposition}
\end{equation}
\end{theorem}

\begin{proof}
Exact ties lie in $\{\Gamma\le\tau\}$. For $\Gamma>\tau$, a flip requires $-\sign(\Psi)R\ge c_Q\Gamma>c_Q\tau$. Since $R_{\mathrm{opp}}\le|R|$, both inequalities follow.
\end{proof}

The first term depends only on the exact geometry of the decisions; the second measures how large the code's calibrated errors are on those decisions. Distance-level summaries average over both. A code with high average fidelity can fail often when many decisions are close calls, and a code with poor average fidelity is harmless when every margin is large. Section~\ref{sec:exp-necessity} constructs two codes with identical distance rank correlations whose flip rates differ by a factor of two.

When close calls thin out polynomially and errors have sub-Gaussian tails, choosing the threshold well turns the bound into an explicit rate.

\begin{corollary}[Rate under sub-Gaussian errors]
\label{cor:subgaussian}
Suppose $\Pp(\Gamma\le\tau)\le C\tau^\beta$ for $0<\tau\le\tau_0$ and $\Pp(|R|\ge z)\le2\exp(-z^2/2\nu^2)$. Here $\nu$ is a tail scale for $R$ itself, not centred, so it absorbs any bias and can exceed $v$; write $r=\nu/c_Q$. Then
\begin{equation}
\Pp(\text{flip})\le\inf_{0<\tau\le\tau_0}\left[C\tau^\beta+2\exp\!\left(-\frac{\tau^2}{2r^2}\right)\right],
\label{eq:small-ball-bound}
\end{equation}
and a threshold of order $r\sqrt{\log(1/r)}$ gives a bound of order $r^\beta[\log(1/r)]^{\beta/2}$ (Appendix~\ref{app:p1}).
\end{corollary}

We call $\beta$ the \emph{near-tie exponent}. It describes the embedding and its decisions rather than the code. Fitting the empirical margin distribution on logarithmic axes gives $\beta=1.53$ on Cohere ($R^2=0.994$), $1.50$ on MiniLM ($R^2=0.979$), and $1.60$ on GIST ($R^2=0.989$), with the same value for one- and two-bit codes on the same data. The close-call term can therefore be estimated from the embedding before a code is chosen, and it indicates how much error control the decisions of an index will require.

\subsection{Ranking: Comparisons That Share a Query}
\label{sec:ranking}

For a query $q$ and candidates $x,y$ with $d(q,x)<d(q,y)$, the ranking score and its comparison error are
\begin{equation}
\Psi_{\mathrm{rank}}=d(q,y)-d(q,x)>0,\qquad R_{\mathrm{rank}}=\xi_y-\xi_x,
\end{equation}
where $\xi_x=\hat d_Q(q,x)-c_Qd(q,x)$ is the calibrated error of the single distance from $q$ to $x$. Both distance errors involve the same quantized query, so they share a component, and the variance of their difference contains a covariance term. For a randomized code and given vectors, this covariance enters the flip probability directly.

\begin{proposition}[Ranking error]
\label{prop:ranking}
Fix $q$, $x$, and $y$, and let the randomness come from the code. Suppose $(\xi_x,\xi_y)$ has mean $m$ and joint sub-Gaussian proxy matrix $K$. Then $R_{\mathrm{rank}}$ is sub-Gaussian with variance proxy
\begin{equation}
\nu_{\mathrm{rank}}^2=K_{xx}+K_{yy}-2K_{xy},
\label{eq:rank-variance}
\end{equation}
and for an exact margin $\gamma>0$ and bias $\mu_R=\E(\xi_y-\xi_x)$,
\begin{equation}
\Pq(\text{flip})\le\exp\!\left[-\frac{(c_Q\gamma+\mu_R)_+^2}{2\nu_{\mathrm{rank}}^2}\right].
\label{eq:ranking-bound}
\end{equation}
\end{proposition}

The cross term $K_{xy}$ is what distance-level summaries leave out. When the two distance errors are positively correlated, the comparison error is smaller than it would be if they were independent.

For a deterministic code, the same structure appears over the decision population. Writing $\widehat\Sigma_\xi$ for the sample covariance of $(\xi_x,\xi_y)$ over the comparisons a search makes, the comparison error satisfies the exact identity
\begin{equation}
\widehat\Var(\xi_y-\xi_x)
=(\widehat\Sigma_\xi)_{xx}+(\widehat\Sigma_\xi)_{yy}
-2(\widehat\Sigma_\xi)_{xy}.
\label{eq:empirical-rank-variance}
\end{equation}
Table~\ref{tab:covariance} measures this covariance on \emph{nearest-vs-rest comparisons}, which compare each query's exact nearest neighbour with every other candidate in its 32-neighbour set. The \emph{cancellation factor} $[(\widehat\Sigma_\xi)_{xx}+(\widehat\Sigma_\xi)_{yy}]/\widehat\Var(\xi_y-\xi_x)$ states how many times smaller the comparison-error variance is than it would be with independent distance errors.

\begin{table}[tbp]
\centering
\small

\begin{tabular}{llcccc}
\toprule
Dataset & Code & Error correlation & Cancellation factor & Flip rate & Check \\
\midrule
Cohere-768 & 1-bit & 0.794 & 4.71 & 8.89\% & pass \\
Cohere-768 & 2-bit & 0.870 & 7.42 & 5.83\% & pass \\
Cohere-768 & rotated 1-bit & 0.465 & 1.87 & 15.58\% & pass \\
Cohere-768 & rotated 2-bit & 0.435 & 1.77 & 8.81\% & pass \\
MiniLM-384 & 1-bit & 0.047 & 1.05 & 11.21\% & pass \\
MiniLM-384 & 2-bit & 0.093 & 1.10 & 4.19\% & pass \\
GIST-960 & 1-bit & 0.902 & 8.89 & 100.00\% & fail \\
GIST-960 & 2-bit & 0.832 & 5.70 & 35.83\% & fail \\
GIST-960 & rotated 1-bit & 0.585 & 2.34 & 14.80\% & pass \\
GIST-960 & rotated 2-bit & 0.529 & 2.09 & 9.33\% & pass \\
\bottomrule
\end{tabular}

\small
\caption{Distance-error correlation and cancellation factor for nearest-vs-rest comparisons on 32-neighbour sets. Calibration is comparison-fitted on 80 queries; flip rates use 160 held-out queries. ``Check'' is the applicability check of Section~\ref{sec:applicability}.}
\label{tab:covariance}
\end{table}

Two observations follow. First, these correlations describe selected nearest-vs-rest comparisons. For i.i.d.\ endpoints and a symmetric distance-error kernel, the Hoeffding decomposition caps shared-endpoint error correlation at $1/2$ and cancellation at $2$ (Appendix~\ref{app:p2p3}). The symmetric-code experiments use the same scorer on both endpoints. Directional query--candidate scorers use the general joint-error formula. Unrotated Cohere codes and GIST exceed both caps. Rotation lowers Cohere's fitted correlation to 0.435 and cancellation to 1.77, while held-out flips rise from 5.83\% to 8.81\%. The calibration cross in Section~\ref{sec:exp-rotation} explains this change through bias, scale, and margins jointly. Second, combining the close-call split with the one-sided tail gives an empirical plug-in value of 0.55 for Cohere two-bit codes. The same calculation without covariance gives one on three Cohere configurations and 0.87 on rotated two-bit codes. Appendix~\ref{app:covariance} specifies the fitted constants, threshold selection, and aggregation. A six-bin margin-dependent scale changes the configuration-level correlation by at most 0.002.

\subsection{Pruning: Comparisons That Share a Candidate}
\label{sec:pruning}

Vamana's pruning rule compares a target $t$, a candidate $c$, and an already selected neighbour $s$ through $\Psi_\alpha=d(t,c)-\alpha\,d(s,c)$. Its comparison error
\begin{equation}
Z_\alpha=\xi_{tc}-\alpha\,\xi_{sc},\qquad
\nu_\alpha^2=K_{tc,tc}+\alpha^2K_{sc,sc}-2\alpha K_{tc,sc},
\end{equation}
combines two distances that share the candidate $c$ rather than a query. A calibration intercept $b$ cancels in a ranking comparison but survives here as $(1-\alpha)b$, so a pruning comparison carries its own bias $\mu_Z=\E Z_\alpha$. Pruning can also go wrong in two ways, and the bound treats them separately.

\begin{proposition}[Pruning error]
\label{prop:pruning}
Fix the triple $(t,c,s)$ and let the randomness come from the code, with $Z_\alpha$ sub-Gaussian with scale $\nu_\alpha$ and mean $\mu_Z$. For an exact margin $\gamma=|\Psi_\alpha|>0$:
\begin{itemize}
\setlength{\itemsep}{2pt}
\setlength{\parsep}{0pt}
\item a \emph{missed prune}, where $\Psi_\alpha>0$ but $c$ is kept, has probability at most $\exp[-(c_Q\gamma+\mu_Z)_+^2/(2\nu_\alpha^2)]$;
\item a \emph{spurious prune}, where $\Psi_\alpha<0$ but $c$ is pruned, has probability at most $\exp[-(c_Q\gamma-\mu_Z)_+^2/(2\nu_\alpha^2)]$;
\item at an \emph{exact tie}, $\Psi_\alpha=0$, the two actions differ exactly when $Z_\alpha<0$, and this probability enters the risk as a separate term.
\end{itemize}
\end{proposition}

The bias decides which error dominates. With the scale-only calibration of Section~\ref{sec:exp-trace}, every observed pruning error is a missed prune; a fitted intercept shifts the bias and produces both kinds.

This is the standard RobustPrune rule analysed here and in Section~\ref{sec:trace}, and replayed on calibrated scores in E3. QuIVer's native builder uses $h(c,t)=4D-\langle q(c),q(t)\rangle$, with $q_i\in\{-2,-1,1,2\}$ and smaller scores nearer, and tests $h(c,s)<\alpha h(c,t)$ and fills the remaining degree slots in distance order, called \emph{backfill}. About one candidate per vertex survives its test, and backfill supplies 98.4\% of final edges in the production Cohere-1M build~\citep{xiao2026quiver}. E8 evaluates this native builder. Its integer predicate and production candidate pools differ from E3's standard rule on fixed 64-candidate pools with degree 32.

\subsection{Top-$K$ Selection}
\label{sec:topk}

For a candidate set $C$ fixed before scoring, with exact top-$K$ subset $A\subset C$, the quantized top-$K$ set differs from $A$ only if some pair with $x\in A$ and $y\in C\setminus A$ is ranked in the wrong order. Hence
\begin{equation}
\Pq(\widehat A\ne A)\le\sum_{x\in A}\sum_{y\in C\setminus A}\Pq(\text{ranking of }(x,y)\text{ flips}),
\label{eq:topk-union}
\end{equation}
and any fixed set of comparisons that determines the top-$K$ set can replace the full cross product.

% ============================================================
\section{When Does a Neighbour List Change?}
\label{sec:trace}

A neighbour-selection call makes dozens of dependent comparisons, yet whether its output changes follows a simple rule: the output changes exactly when the decision on one candidate, replayed on the exact run, changes.

\subsection{Neighbour Selection as a Sequence of Actions}
\label{sec:selection}

Fix a target $t$, a candidate order $\pi=(c_{(1)},\ldots,c_{(N)})$ sorted by exact distance to $t$ with a fixed tie-breaking rule, a degree budget $\Delta$, and $\alpha\ge1$. Selection starts with an empty list $S$ and scans $\pi$ until the order is exhausted or $|S|=\Delta$. For each candidate $c$ it computes the action
\begin{equation}
a(c;S)=\bigvee_{s\in S}\mathbf 1\{d(t,c)-\alpha\,d(s,c)\ge0\},
\end{equation}
and appends $c$ exactly when $a(c;S)=0$. A selected neighbour $s$ with $d(t,c)-\alpha\,d(s,c)\ge0$ is a \emph{blocking neighbour} of $c$; a candidate is appended when no selected neighbour blocks it. This is the standard RobustPrune call without backfill. Backfill is a separate map $B_\pi$ that fills remaining slots in the shared order $\pi$. Applying the same map to equal plain outputs gives equal backfilled outputs (Appendix~\ref{app:p4}).

The \emph{decision record} of a call lists the candidate order, each visited candidate, the selected list before it, its action, and where the scan stopped. It records which candidates are blocked, and leaves out the order in which an implementation checks their blocking neighbours (Figure~\ref{fig:first-divergence}).

\begin{figure}[tbp]
\centering
\includegraphics[width=\textwidth]{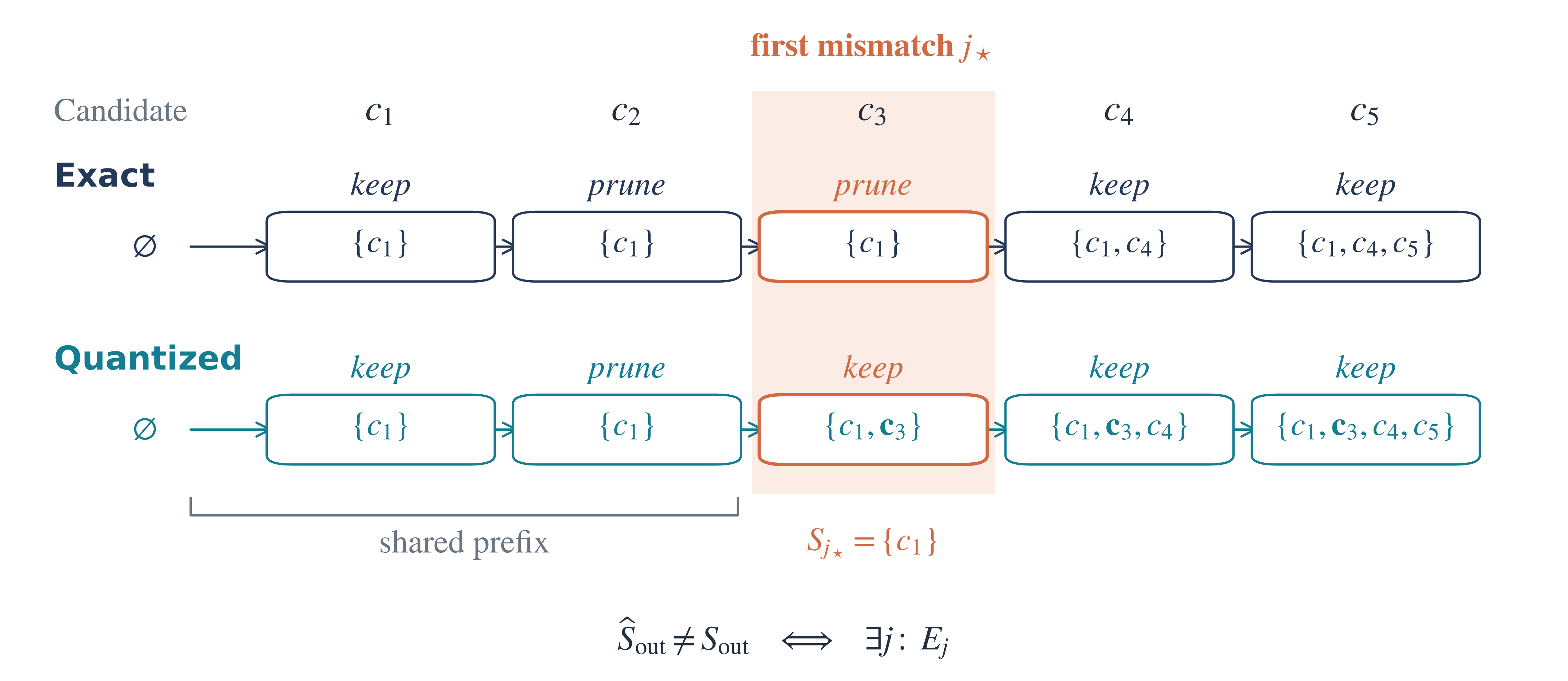}
\small
\caption{First divergence under a shared candidate order, $\Delta=4$, and no backfill. Each box shows the list after its action. Both calls reach $c_3$ with $\{c_1\}$, but only the quantized call appends it. Later boxes follow the actual runs. Theorem~\ref{thm:trace} evaluates replayed actions on exact-run lists.}
\label{fig:first-divergence}
\end{figure}

\subsection{The First-Divergence Theorem}
\label{sec:first-divergence}

Replay the quantized call with the exact call's candidate order and, at each step, the exact call's selected list. If every replayed action agrees with the exact one, the two calls make the same choices throughout; otherwise the first disagreement is where their outputs part.

\begin{theorem}[First divergence]
\label{thm:trace}
Let the quantized replay use the same candidate order, distinct candidate labels, degree budget, and tie rule as the exact call. For each candidate $c_j$ visited by the exact call, evaluate the quantized action on the exact selected list $S_j$, and let $E_j$ be the event of an \emph{action mismatch}, that this action differs from the exact one. Then
\begin{equation}
\{\widehat T\ne T\}=\{\widehat S_{\mathrm{out}}\ne S_{\mathrm{out}}\}=\bigcup_{j\in\mathcal V}E_j,
\label{eq:trace-identity}
\end{equation}
where $T,\widehat T$ are decision records, $S_{\mathrm{out}},\widehat S_{\mathrm{out}}$ are output lists, and $\mathcal V$ is the exact visited set. With the exact input and order fixed, write $p_{j,s}\in[0,1]$ for a bound on the flip probability of $(t,c_j,s)$ and $B_j=\{s\in S_j:\Psi_{j,s}\ge0\}$ for the exact blockers. Then
\begin{equation}
q_j=\begin{cases}
\min_{s\in B_j}p_{j,s},&B_j\ne\varnothing,\\
\min\{1,\sum_{s\in S_j}p_{j,s}\},&B_j=\varnothing
\end{cases}
\label{eq:action-bound}
\end{equation}
bounds $\Pp(E_j)$, with an empty sum equal to zero. Consequently,
\begin{equation}
\Pp(\widehat S_{\mathrm{out}}\ne S_{\mathrm{out}})
\le\sum_{j\in\mathcal V}q_j
\le\sum_{j\in\mathcal V}\sum_{s\in S_j}p_{j,s}.
\label{eq:trace-bound}
\end{equation}
\end{theorem}

\begin{proof}
If no $E_j$ occurs, induction from the empty list makes the selected lists, actions, and stopping points identical. Otherwise, at the first mismatch the candidate is appended in exactly one call and remains there, so the outputs differ. If $B_j$ is nonempty, an action mismatch requires every exact blocker to fail, giving the minimum in~\eqref{eq:action-bound}. If $B_j$ is empty, a union over spurious blockers gives the sum.
\end{proof}

The identity is deterministic. Probability enters only through the bounds on single comparisons, and the theorem uses no independence between comparisons. When quantization also changes the candidate order, the probability that the order changes is added to the bound, and the theorem applies on the event that the orders agree.

\subsection{Edges and Paths}
\label{sec:edges}

Many applications need less than a whole list. For a single selected edge, it suffices that every action before it agrees.

\begin{corollary}[Edge and path survival]
\label{cor:path}
Let $e$ be an exact output edge selected at scan position $j_e$. If every action at positions up to $j_e$ agrees, $e$ survives in the quantized output. For a fixed path $P$ made of such edges,
\begin{equation}
\Pp(\text{some edge of }P\text{ is lost})\le\sum_{e\in P}\sum_{j\le j_e}q_{e,j}.
\end{equation}
Here each edge uses its own exact-run prefixes, and repeated action events may be counted once. If the $N_P$ retained action bounds satisfy $q_{e,j}\le\exp(-\eta_{\min}^2/2)$, their sum is below one whenever $\eta_{\min}^2>2\log N_P$.
\end{corollary}

In the calls of Section~\ref{sec:exp-trace}, an edge prefix contains 21 to 25 actions on average. The blocker-aware plug-in sums exceed one for 89 to 92\% of edges and correlate with edge loss at 0.81, 0.69, and 0.73 on Cohere, MiniLM, and GIST. Appendix~\ref{app:trace} compares them with the all-comparisons sums.

\begin{remark}[The largest risks predict failure]
\label{rem:topr}
\begingroup\emergencystretch=2em
For the failure of a query's candidate set, the largest ranking-comparison risks carry most of the predictive value: ranking the 31 comparisons of a candidate set by individual risk and summing only the $r$ largest raises the correlation with observed failure as $r$ grows to 3 to 5 (Cohere two-bit from 0.471 to 0.480, MiniLM one-bit from 0.460 to 0.500; MiniLM rotated two-bit reaches 0.589), after which further terms add noise. The full sum is the upper bound, and the sum of the largest few terms is the better predictor.
\par\endgroup
\end{remark}

% ============================================================
\section{Where the Error Law Comes From}
\label{sec:bridge}

The close-call bound needs the probability of a large comparison error. For a deterministic code, this probability is a property of the data: it depends on how the code transforms the joint distribution of the embedding's coordinates. This section derives that error law for coordinate sign and two-bit codes under a Gaussian model of the embedding. It has two parts: the signal a code keeps, which sets the calibration slope (Sections~\ref{sec:stein}--\ref{sec:rotation}), and the error it adds, whose variance and tails enter the bound (Sections~\ref{sec:error-decomposition}--\ref{sec:routes}). The section also shows which checks of the model suffice to trust its error tails and which do not.

\subsection{A Gaussian Model of Contrastive Embeddings}
\label{sec:gaussian-model}

Let $G\in\R^D$ be the encoder output before $L_2$ normalization.

\begin{assumption}[Coordinate Gaussianity]
\label{asm:H1}
\begin{enumerate}
\setlength{\itemsep}{2pt}
\setlength{\parsep}{0pt}
\item[(a)] For every fixed set $I$ of $k$ coordinates, $d_{\mathrm{BL}}(\mathcal L(G_I),\mathcal L(Z_I))\le\varepsilon_D$ with $Z\sim\mathcal N(\mu,\Sigma)$.
\item[(b)] $\E\bigl|\|G\|-r_D\bigr|/r_D\le\delta_D$ with $\delta_D\to0$.
\end{enumerate}
\end{assumption}

Part (a) says that small groups of coordinates are jointly close to Gaussian in the bounded-Lipschitz distance $d_{\mathrm{BL}}$. Part (b), a \emph{thin shell}, says that the norm of $G$ varies little relative to its typical size $r_D$. \citet{betser2026infonce} prove that fixed-dimensional projections of InfoNCE-trained representations become Gaussian under stated alignment and concentration assumptions. Direct coordinate-wise QQ fits on archived Cohere, MiniLM, and Wolt vectors have mean $R^2$ between 0.9969 and 0.9993, against 0.9181 for GIST (Appendix~\ref{app:foundational}). Joint Gaussianity is stronger, and Section~\ref{sec:exp-eligibility} examines it separately. In the \emph{exact Gaussian model}, $G\sim\mathcal N(\mu,\Sigma)$, the following identities hold exactly.

\subsection{How a Sign Code Inherits Coordinate Correlations}
\label{sec:stein}

The first quantity a code determines is how much of the exact score its own score retains. For a sign code and a linear score, Stein's lemma gives this covariance exactly. Let $G\sim\mathcal N(\mu,\Sigma)$, let $S_0=\sum_i\sign(G_i)$ be the sum of the sign bits, and let $T_0=\sum_jw_jG_j$ be a linear score with fixed weights $w_j$.

\begin{theorem}[Stein covariance identity]
\label{thm:stein}
\begin{equation}
\Cov(S_0,T_0)=\sum_{i,j}a_i\,\Sigma_{ij}\,w_j,\qquad a_i=\frac{2\varphi(\mu_i/\sigma_i)}{\sigma_i},
\end{equation}
where $\sigma_i^2=\Sigma_{ii}$ and $\varphi$ is the standard normal density.
\end{theorem}

\begin{proof}
By bilinearity $\Cov(S_0,T_0)=\sum_{i,j}w_j\Cov(\sign(G_i),G_j)$. For jointly Gaussian $(G_i,G_j)$ the conditional mean is $\E[G_j\mid G_i]=\mu_j+(\Sigma_{ij}/\sigma_i^2)(G_i-\mu_i)$, whence
\[
\Cov(\sign(G_i),G_j)=\frac{\Sigma_{ij}}{\sigma_i^2}\cdot2\sigma_i\varphi(\mu_i/\sigma_i)=a_i\Sigma_{ij}.
\]
\end{proof}

This identity is the Gaussian input--output cross-covariance relation of \citet{bussgang1952crosscorrelation}. Every covariance $\Sigma_{ij}$ enters, not only the variances. Individual coordinate correlations are small on our data, with mean $|\rho_{ij}|$ between 0.04 and 0.11, but there are $D(D-1)$ of them: with $|\rho_{ij}|\asymp\kappa/\sqrt D$, their squared sum $\sum_{i\ne j}\rho_{ij}^2\asymp\kappa^2D$ is of the same order as the diagonal. Measured by the energy $\mathcal I_{\mathrm{off}}=\sum_{i\ne j}(a_i\Sigma_{ij})^2$, the off-diagonal terms carry 31 to 38\% of the predicted signal. Table~\ref{tab:fidelity} tests the model on a quantity it predicts directly, the ranking fidelity of one-bit codes. With the full covariance, the model reproduces the measured fidelity on four datasets; with only the diagonal, it falls short by 0.21 to 0.35. Fidelity serves here as a test of the model, not as a measure of decision quality.

\begin{table}[tbp]
\centering

\begin{tabular}{lcccc}
\toprule
Dataset & $F_{\mathrm{measured}}$ & $F_{\mathrm{full}}$ & $F_{\mathrm{diag}}$ & Gap explained \\
\midrule
Cohere & 0.681 & 0.688 & 0.474 & 103\% \\
Arxiv & 0.897 & 0.886 & 0.546 & 97\% \\
CodeSearch & 0.823 & 0.837 & 0.542 & 105\% \\
Random & 0.907 & 0.899 & 0.560 & 98\% \\
\bottomrule
\end{tabular}

\small
\caption{Spearman ranking fidelity $F$ of one-bit codes: measured, and predicted by the Gaussian model with the full covariance and with its diagonal only. Gap explained is $(F_{\mathrm{full}}-F_{\mathrm{diag}})/(F_{\mathrm{measured}}-F_{\mathrm{diag}})$.}
\label{tab:fidelity}
\end{table}

\subsection{Why a Magnitude Bit Helps}
\label{sec:magnitude}

QuIVer's code adds a second, magnitude bit to each sign. In its fixed-amplitude form, coordinate $i$ of a vector $g$ is coded as $q_i(g)=\sign(g_i)(1+\mathbf 1\{|g_i|>\tau\})$, which takes the values $\pm1$ and $\pm2$. Whether the extra bit helps depends on what the score is meant to track. In the \emph{matched-scorer model}, $G$ and $H$ are independent with independent coordinates $G_i,H_i\sim\mathcal N(0,\sigma_i^2)$, the target is their inner product, and the scorers sum per-coordinate products of the codes:
\begin{equation}
T=\sum_iG_iH_i,\qquad S_1=\sum_i\sign(G_i)\sign(H_i),\qquad S_2=\sum_iq_i(G)q_i(H).
\end{equation}
Target and scorers thus weight the coordinates alike.

\begin{proposition}[Magnitude-bit gain]
\label{prop:magnitude}
For every $D\ge1$, every $\sigma_i>0$, and every finite threshold $\tau>0$, the Pearson correlations satisfy $\rho(S_2,T)>\rho(S_1,T)$.
\end{proposition}

The closed forms and the proof are in Appendix~\ref{app:p5}. The match between scorer and target is what makes the inequality hold. When the target is a linear score that weights the coordinates differently from a fixed equally weighted scorer, the second bit can lower the correlation; in the eight-dimensional example of Appendix~\ref{app:p5} it lowers it by 0.163. A second bit gives the best possible readout more information, and a fixed readout gains when it weights the coordinates as the target does. With QuIVer's per-vector threshold, the mean absolute coordinate, in place of the fixed $\tau$, the two correlations converge to their fixed-threshold values as the number of independent equal-variance coordinates grows.

On the contrastive embeddings we inspected, the second bit raises the rank correlation between quantized and exact distances by 0.071 to 0.132. On Cohere it also brings the probability of a comparison error beyond three error scales from 2.09 times the Gaussian value to 1.05 times, and it lowers the flip rate of nearest-vs-rest comparisons by about a third, from 8.89\% to 5.83\%.

\subsection{Two Uses of Random Rotation}
\label{sec:rotation}

A random orthogonal rotation changes a representation's coordinate structure and nothing else, which makes it the cleanest intervention on a coordinate code. By L\'evy concentration and a union bound, a Haar-random rotation makes every coordinate variance, with high probability, equal to
\[
\frac{\tr(\Sigma)}{D}\pm O\!\left(\|\Sigma\|_{\mathrm{op}}\sqrt{\frac{\log D}{D}}\right).
\]
This one fact has opposite consequences for the two families of codes. A coordinate code such as QuIVer's exploits the differences between coordinates, in their variances and in how often each is positive, and rotation removes those differences; its effect on a coordinate code therefore depends on the representation. A randomized code such as RaBitQ relies on exactly the uniformity that rotation creates.

The \emph{sign entropy} $H_{\mathrm{sign}}$ is the average, over coordinates, of the binary entropy of the probability that a coordinate is positive. It equals one when every coordinate is positive and negative equally often, and zero when every coordinate has a fixed sign. Across 12 datasets, the entropy gap $1-H_{\mathrm{sign}}$ predicts how much rotation changes the two-bit code's distance rank correlation with Spearman correlation 0.909, against 0.657 for the spread of the coordinate standard deviations, $\mathrm{CV}(\sigma)$; their product $(1-H_{\mathrm{sign}})\times\mathrm{CV}(\sigma)$ reaches 0.930. On GIST, where each coordinate almost always has the same sign ($H_{\mathrm{sign}}=0.0004$), rotation turns an unusable code into a usable one (Section~\ref{sec:exp-eligibility}). Rotation also changes how the two errors of a comparison are correlated, which the distance rank correlation does not see; on Cohere this change can outweigh the gain in distance accuracy (Section~\ref{sec:exp-rotation}).

\subsection{Decomposing the Error}
\label{sec:error-decomposition}

The comparison error has a part the model predicts exactly and two parts that measure departure from it. The error depends on two per-vector quantities: the magnitude threshold that the two-bit code compares each coordinate with, and the norm by which cosine distance normalizes each vector. An idealized version replaces both by population values: the threshold by its mean over the data, and the norm by the typical norm $r_D$. For a fixed decision type, ranking or pruning, the comparison error then splits as
\begin{equation}
R=R^\circ+R_{\mathrm{thr}}+R_{\mathrm{norm}},
\label{eq:population-decomposition}
\end{equation}
where $R^\circ$ is the \emph{idealized error}, $R_{\mathrm{thr}}$ the \emph{threshold error} from using each vector's own threshold, and $R_{\mathrm{norm}}$ the \emph{normalization error} from using each vector's own norm.

For nodes drawn independently from the model, the idealized error has an exact variance formula. Write the comparison as a combination of node pairs with coefficients $A_p$; a ranking comparison, for example, has $A=+1$ on the pair $(q,y)$ and $A=-1$ on $(q,x)$. Let $d_u=\sum_{p\ni u}A_p$ be the net coefficient of node $u$. Then
\begin{equation}
\Var(R^\circ)=\kappa_1\sum_ud_u^2+\kappa_2\sum_pA_p^2,
\label{eq:oracle-variance}
\end{equation}
where $\kappa_1$ and $\kappa_2$ are the variances of the first- and second-order Hoeffding components of the error of a single pair (Appendix~\ref{app:p5}). In a ranking comparison the query's net coefficient is zero, so its first-order contribution cancels; this is the cancellation of Section~\ref{sec:ranking} for independently drawn nodes, and selection by a search adds to it (Section~\ref{sec:exp-covariance}). The Efron--Stein replacement bound $V_{\mathrm{ES}}$ on this variance satisfies $\Var(R^\circ)\le V_{\mathrm{ES}}\le2\Var(R^\circ)$, with equality at two when the first-order component vanishes.

The threshold and normalization errors are measurable, and their size shows whether the model describes an embedding. The threshold error accounts for 1.3 to 4.7\% of the error variance on Cohere and MiniLM, for essentially all of it on unrotated GIST, and for 2.3\% on GIST after rotation.

\subsection{A Rare-Spike Tail Obstruction}
\label{sec:counterexample}

The population-radius error $R^\circ$ of Equation~\eqref{eq:population-decomposition} admits an explicit rare-spike construction. It has nearly Gaussian low-dimensional marginals, concentrated norms, and isotropic covariance, while its idealized comparison error has a large fourth cumulant.

\begin{theorem}[Tail obstruction for the population-radius error]
\label{thm:counterexample}
There is a sequence of distributions $G_D$ such that the standardized marginals of every fixed set of $k$ coordinates converge to Gaussian in total variation at rate $O_k(D^{-1})$, the relative mean-square deviation of the norm is $O(D^{-1})$, and the covariance is isotropic. With population magnitude threshold $\tau_D=\E|G_{D1}|$, the expected fraction of coordinates satisfying $\bigl||G_{Di}|-\tau_D\bigr|\le D^{-1/4}$ is $O(D^{-1/4})$. The idealized ranking error $R_D^\circ$ with this fixed threshold and population radius satisfies
\begin{equation}
V_D=\Var(R_D^\circ)=\Theta(D^{-1}),\qquad \kappa_4(R_D^\circ)=\Omega(D^2).
\end{equation}
Consequently, a Bernstein-type cumulant condition, or a two-sided sub-gamma bound with variance proxy $v_D=O(V_D)$, requires a scale parameter with $b_D/\sqrt{V_D}=\Omega(D^2)$.
\end{theorem}

\begin{proof}[Construction]
Take $G_D=(1-J_D)Z_D+J_DD^{3/2}S_D$ with $J_D\sim\mathrm{Bernoulli}(D^{-4})$, $Z_D\sim\mathcal N(0,I_D)$, and $S_D$ a vector of independent random signs. The event that two of the three nodes of a ranking comparison are spikes and the third is not has probability $\Theta(D^{-8})$, and on this event the bilinear part of the error contributes $\Theta(D^2)$ to the fourth moment. The bounded-code part has fourth moment $O(D^{-2})$ and cannot cancel this contribution, while the total variance stays $\Theta(D^{-1})$. Appendix~\ref{app:p8} verifies the stated properties and derives the cumulant and moment-generating-function consequences.
\end{proof}

A sub-gamma bound at variance scale $1/D$ therefore needs a scale parameter $D^2$ times the standard deviation. This error uses the common population radius $r_D$. Per-vector cosine normalization removes the spike amplitude and enters the separate term $R_{\mathrm{norm}}$ in Equation~\eqref{eq:population-decomposition}. Tail control for $R^\circ$ can instead bound coordinate influence, conditional error increments, or the range after rare events are clipped.

\subsection{Three Routes to Light Tails}
\label{sec:routes}

The last two controls lead to two analytical routes; the third route measures the tail directly.
\begin{enumerate}
\setlength{\itemsep}{2pt}
\setlength{\parsep}{0pt}
\item \textbf{Conditional tail route.} Reveal the three nodes of a comparison one at a time. If the three resulting martingale increments of $R^\circ$ satisfy sub-gamma bounds with parameters $(v_j,b_j)$ conditionally on the earlier nodes, almost surely, then
\begin{equation}
\Pp(|R^\circ-\E R^\circ|\ge z)\le2\exp\!\left[-\frac{z^2}{2(v_\star+b_\star z)}\right],\qquad v_\star=\sum_jv_j,\ b_\star=\max_jb_j.
\label{eq:conditional-mgf-tail}
\end{equation}
\item \textbf{Clipping route.} If a clipped error $R_L$ equals $R^\circ$ outside an event of probability $\varepsilon_L$, has bias $|\E R^\circ-\E R_L|\le d_L$, and is sub-gamma with parameters $(v_L,b_L)$, then
\begin{equation}
\Pp(|R^\circ-\E R^\circ|\ge z+d_L)\le2\exp\!\left[-\frac{z^2}{2(v_L+b_Lz)}\right]+\varepsilon_L.
\end{equation}
\item \textbf{Held-out route.} Section~\ref{sec:heldout} bounds the failure rate directly from independent data, with no tail model.
\end{enumerate}
Section~\ref{sec:exp-eligibility} measures the parameters of the first two routes on real embeddings.

% ============================================================
\section{Measuring Risk Without a Model}
\label{sec:heldout}

When no model describes the comparison error, because the coordinates are far from Gaussian, the error tails are heavy, or the calibration slope is near zero, independent held-out data still bound the failure rate of the decisions a code accepts. Two ideas make this possible. Decisions come in \emph{blocks}, each consisting of one query or target with its candidate set, and independent blocks allow arbitrary dependence among the decisions inside a block. And a certificate covers only the decisions that an acceptance rule, fixed in advance, accepts, so a code can be certified on the decisions it handles reliably.

\subsection{When to Trust a Calibrated Error}
\label{sec:applicability}

The close-call bound holds for every positive signal scale. For empirical use of the analytical error models, we apply a prespecified \emph{heuristic diagnostic screen}. It requires a positive fitted slope, a score rank correlation of at least 0.10, an absolute error--margin rank correlation of at most 0.50, and left-tail probabilities at most twice the Gaussian values at 1 to 2.5 error scales. These checks summarize measured fit and tails. The analytical inequalities retain their stated moment-generating-function assumptions. The covariance study checks all four criteria. Appendix~\ref{app:ranking} reports the first three because its local-decision protocol does not estimate the tail diagnostic. Across 27 nearby choices of the three numerical cutoffs, unrotated GIST fails every screen and MiniLM passes every screen, while some Cohere and rotated-GIST outcomes depend on the tail cutoff (Appendix~\ref{app:eligibility}). Held-out certificates apply to every code.

\subsection{Certificates for Accepted Decisions}
\label{sec:certificates}

\begin{theorem}[Held-out certificate for accepted decisions]
\label{thm:heldout}
Fit all calibration constants, acceptance rules, thresholds, and tie rules on an independent fitting split, and condition on them. Let $Z_1,\ldots,Z_n$ be i.i.d.\ blocks. For an acceptance rule $s$ and a threshold $\tau$, let $C_{i,s}$ be the fraction of decisions in block $i$ that $s$ accepts, $U_{i,s,\tau}$ the fraction that are accepted and are either close calls ($\Gamma\le\tau$, including exact ties) or large errors ($|R|\ge c_Q\tau$), and $F_{i,s}$ the fraction that are accepted and flip, so that
\begin{equation}
0\le F_{i,s}\le U_{i,s,\tau}\le C_{i,s}\le1.
\end{equation}
For a prespecified grid $\mathcal G$ of pairs $(s,\tau)$, set $\epsilon_n=\sqrt{\log(|\mathcal G|/\delta)/(2n)}$. With probability at least $1-\delta$, simultaneously for every pair with $\widehat C_s>0$,
\begin{equation}
p_s:=\frac{\E F_{i,s}}{\E C_{i,s}}
\;\le\;
q_{s,\tau}:=\frac{\E U_{i,s,\tau}}{\E C_{i,s}}
\;\le\;
\min\!\left\{1,\frac{\widehat U_{s,\tau}+\epsilon_n}{\widehat C_s}\right\},
\label{eq:certificate}
\end{equation}
where hats denote averages over the $n$ blocks. The ratios are defined when $\E C_{i,s}>0$, and a rule that accepts no held-out decision yields no certificate.
\end{theorem}

\begin{proof}
Fix a grid pair and write $q=q_{s,\tau}$. Since $0\le U_i\le C_i\le1$, the variable $W_i=U_i-qC_i$ has mean zero and lies in $[-q,1-q]$, an interval of length one. One-sided Hoeffding gives $\widehat U-q\widehat C\ge-\epsilon_n$ except with probability $\delta/|\mathcal G|$; a union bound over the grid and division by $\widehat C>0$ complete the proof. Decisions inside a block may depend on each other arbitrarily.
\end{proof}

\begin{corollary}[Failure certificate]
\label{cor:direct-failure}
Replacing $U_{i,s,\tau}$ by $F_{i,s}$ and paying only for the family $\mathcal S$ of acceptance rules gives, simultaneously over $s\in\mathcal S$,
\begin{equation}
p_s\le\min\!\left\{1,\frac{\widehat F_s+\sqrt{\log(|\mathcal S|/\delta)/(2n)}}{\widehat C_s}\right\}.
\end{equation}
\end{corollary}

The quantity $p_s$ is the failure rate among accepted decisions when a block is drawn at random and then a decision at random within it; pooling all decisions instead would give larger blocks more weight. Both results bound $p_s$. The \emph{mechanism certificate} of Theorem~\ref{thm:heldout} goes through the close-call-or-large-error event, so it keeps the mechanism of Theorem~\ref{thm:main} visible. The \emph{failure certificate} of Corollary~\ref{cor:direct-failure} bounds observed flips directly and is tighter, especially when the rule accepts many decisions. Appendix~\ref{app:p6-selective} gives two further variants: one that bounds the numerator and the denominator separately, which is looser, and a variance-adaptive empirical-Bernstein version.

Section~\ref{sec:exp-certificates} evaluates acceptance rules based on exact MNR and on quantized margins. Each rule is fixed on the fitting split. Exact margins support an offline workload audit, while quantized-margin rules use observable scores. Each certificate describes its block distribution.

\noindent\emph{Protocol.}\quad
The sample-size study of Section~\ref{sec:exp-certificates} fits on an independent split, certifies on the first 128, 256, 512, and 1{,}024 of a sequence of i.i.d.\ blocks, and validates on 1{,}024 further blocks. Each block independently samples one query or target and 32 distinct candidates from its fixed pool. Vectors may recur across blocks. Every method and sample size is fixed in advance, so each reported certificate carries its own $1-\delta$ guarantee. Appendix~\ref{app:heldout} reports a separate five-pool protocol with its own fit and seeds.

% ============================================================
\section{Other Quantizers}
\label{sec:instantiations}

The close-call bound and the first-divergence theorem apply to any code; what changes from one code to another is the mechanism behind the comparison error. Table~\ref{tab:cross-quantizer} places six reference codes on this common interface.

\emph{Coordinate codes} are the family analysed in Section~\ref{sec:bridge}; the Stein identity, the magnitude-bit result, the effect of rotation, and the conditional tail route are specific to it. \emph{RaBitQ}~\citep{gao2023rabitq} is a randomized ratio estimator with an error radius for each rotation. For a fixed query $q$ and data directions $x$ and $y$ that share one random rotation,
\begin{equation}
\Pp\{|\varepsilon_R(x,q)-\varepsilon_R(y,q)|>\rho_x+\rho_y\}\le4\exp(-c_0\epsilon_0^2),
\end{equation}
where $\rho_x$ and $\rho_y$ are the random radii and $c_0,\epsilon_0$ are constants of the estimator. Because the two distances share the rotation, ranking and pruning bounds go through this joint event and the triangle inequality rather than through independence. A pruning rule stated in Euclidean distance becomes an $\alpha^2$ rule on squared distances, with a different radial factor for each of the two distances (Appendix~\ref{app:p7}). \emph{Lucene BBQ}~\citep{trent2024bbq} stores one-bit vectors and scores them against int4 query-side vectors. During HNSW construction it keeps temporary int4 query-side vectors, scores candidates against the stored one-bit vectors, and uses the int4 vectors for diversity and reverse-link scoring, so each scoring direction $A\to B$ needs its own affine calibration $\widehat d_{A\to B}=b_{A\to B}+a_{A\to B}\,d+\xi_{A\to B}$. \emph{Product quantization}~\citep{jegou2011product,ge2014optimized} decomposes the error over codebook blocks, which enter through bounded-block concentration or held-out tails; the coordinate-sign Stein identity does not apply to it.

The RaBitQ reference uses a seeded signed-DCT orthogonal transform without scalar quantization of the query; the BBQ-like reference uses centroid-centred one-bit codes and int4 queries and leaves out Lucene's metric-specific corrections; the product quantizer uses 16 subspaces with 16 centroids each.

\begin{table}[tbp]
\centering

\small
\begin{tabular}{lrrrrr}
\toprule
& & \multicolumn{2}{c}{Ranking flips} & \multicolumn{2}{c}{Pruning flips} \\
Code & Distance corr. & all & hard & all & hard \\
\midrule
Sign (1-bit) & .373 & 35.9\% & 47.1\% & 20.1\% & 36.5\% \\
Shared 2-bit & .789 & 18.5\% & 41.5\% & 12.8\% & 33.9\% \\
RaBitQ reference & .939 & 11.8\% & 37.2\% & 6.5\% & 27.1\% \\
BBQ-like & .923 & 12.8\% & 37.3\% & 9.0\% & 30.7\% \\
Scalar int4 & .997 & 2.1\% & 10.1\% & 1.2\% & 5.8\% \\
PQ (16 blocks $\times$ 4 bits) & .801 & 22.4\% & 43.3\% & 11.0\% & 33.1\% \\
\bottomrule
\end{tabular}

\small
\caption{Six codes on the common interface, averaged over Cohere, MiniLM, SIFT, and GIST with 60{,}000 held-out ranking and 60{,}000 pruning comparisons per dataset. ``Distance corr.'' is the correlation between quantized and exact distances; ``hard'' restricts to the 20\% of comparisons with the smallest exact margins. All calibration slopes are positive. The product quantizer stores 64 index bits per vector.}
\label{tab:cross-quantizer}
\end{table}

Scalar int4 has the best distance-level metrics of these codes, with distance correlation 0.997 and 2.1\% ranking flips, yet it still flips 10.1\% of the ranking comparisons in the hardest fifth. More bits shrink the error; the margins of the hardest comparisons stay as small as before.

% ============================================================
\section{Experiments}
\label{sec:experiments}

The data span contrastive text embeddings (Cohere, MiniLM, BGE-M3, Jina, MSMARCO), learned vision embeddings (Wolt-CLIP, Landmark-DINO), word embeddings and visual descriptors (GloVe, SIFT, GIST), and synthetic controls. Decision and graph experiments use $L_2$-normalized vectors with the stated disjoint splits or independent block streams. Representation diagnostics state input scales separately. Table~\ref{tab:protocols} and Appendix~\ref{app:experiments} give the protocols and evidence.

\begin{table}[tbp]
\centering

\small
\setlength{\tabcolsep}{3pt}
\begin{tabular}{L{2.6cm}L{3.3cm}L{4.3cm}L{3.6cm}}
\toprule
Experiment & Data and codes & Decisions & Splits \\
\midrule
E1 margin bins & Cohere, MiniLM, GIST; 1- and 2-bit & nearest-vs-rest on 32-neighbour sets & 80 calibration, 160 held-out queries \\
E1 predictor & 12 datasets, 2-bit rotation pairs; four datasets, six codes & shared 48-neighbour pools; pruning $\Delta=24$ & disjoint calibration / validation / stress; six-code fit/cal/eval 1{,}200/800/2{,}000 \\
E2 correlated errors & Cohere, MiniLM, GIST & nearest-vs-rest and four selection regimes & 4{,}000 or 6{,}000 vectors; 80 or 160 query blocks \\
E3 first divergence & Cohere, MiniLM, GIST; 2-bit with and without rotation & RobustPrune, 64 candidates, $\alpha=1.2$, $\Delta=32$ & 80 calibration, 160 evaluation targets \\
E4 certificates & 12 datasets; 2-bit with and without rotation & 32-candidate blocks & certification prefixes 128--1{,}024; 1{,}024 validation blocks \\
E5 rotation & Cohere & E1 and E2 protocols & as in E1 and E2 \\
E6 Gaussian model & three to eight datasets per check & ranking and pruning comparisons & disjoint fit, reference, calibration, validation pools \\
E7 necessity & synthetic scores & 500{,}000 comparisons & --- \\
E8 graph search & Cohere and Wolt, 200K each & 2 orders $\times$ 2 axes $\times$ 3 widths $\times$ 4 scorers & 128 fitting, 616 evaluation queries \\
\S\ref{sec:instantiations} & Cohere, MiniLM, SIFT, GIST; six codes & 60{,}000 ranking and 60{,}000 pruning & 1{,}200 fit, 2{,}000 held-out vectors \\
\bottomrule
\end{tabular}

\small
\caption{Experimental protocols. Fitting and evaluation follow the stated split or independent-block sampling protocol.}
\label{tab:protocols}
\end{table}

\subsection{Margins and Flip Rates (E1)}
\label{sec:exp-margins}

On Cohere two-bit codes, the flip rate falls steadily with the margin-to-noise ratio over more than two orders of magnitude, and no comparison with a ratio above 5 flipped (Figure~\ref{fig:margin-bins}). On Cohere and MiniLM two-bit codes, two thirds to three quarters of comparisons have a ratio of at least 2, and 7 to 9\% fall below 1. On GIST, which fails the applicability check, the flip rate stays between 28\% and 64\% at every ratio (Appendix~\ref{app:ranking}), so no cutoff separates safe comparisons from unsafe ones.

\begin{figure}[tbp]
\centering
\includegraphics[width=\textwidth]{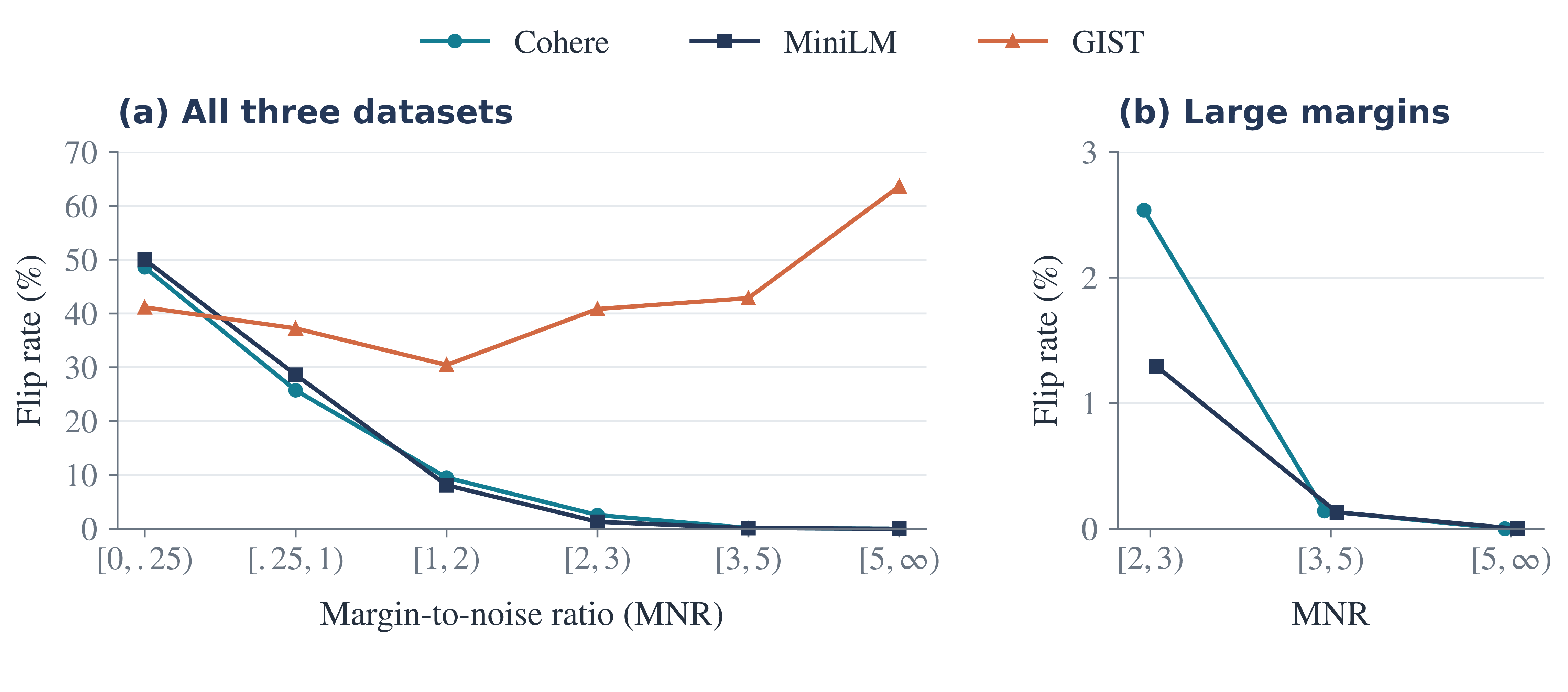}
\small
\caption{Flip rate by MNR for nearest-vs-rest comparisons, with unrotated two-bit codes and $4{,}960$ held-out comparisons per data set. Points pool counts within the six intervals. The right panel enlarges Cohere and MiniLM at large margins. Table~\ref{tab:app-margin-bins} gives the original nine-bin counts.}
\label{fig:margin-bins}
\end{figure}

E1 fits $\widehat d=b+a d_{\mathrm{raw}}$ on the 48 calibration target--candidate blocks and sets $c_Q=1$. For signed comparisons $\Psi$ and residuals $R=\widehat\Psi-\Psi$, the original score is
\begin{equation}
\widehat r_2=\min\!\left\{1,\min_{\tau\in\mathcal T}
\left[\widehat\Pp_{\mathrm{fit}}(\Gamma\le\tau)
+\widehat\Pp_{\mathrm{fit}}(|R|\ge\tau)\right]\right\},
\label{eq:empirical-risk}
\end{equation}
where $\mathcal T$ contains nine fitting-margin quantiles. Replacing $|R|$ by $R_{\mathrm{opp}}$ defines $\widehat r_{\mathrm{opp}}$. Both use empirical frequencies and fitting data only. Appendix~\ref{app:calibration-direction} specifies the grid, orientations, and positive comparison-fitting control. Across 24 configurations, $\widehat r_2$ correlates with held-out flips at 0.970 for ranking and 0.992 for pruning (Figure~\ref{fig:predictor-scatter}). Distance rank correlation gives 0.739 / 0.053 and distance MSE 0.184 / 0.447. The scores exceed observed flips by 24--25 percentage points on average. Excluding the nonpositive-slope configuration gives 0.966 / 0.991. Dataset-cluster leave-one-out ranges are 0.961--0.981 / 0.990--0.994, and cluster-bootstrap 95\% intervals are $[0.850,0.992]$ / $[0.955,1.000]$.

\begin{figure}[tbp]
\centering
\includegraphics[width=\textwidth]{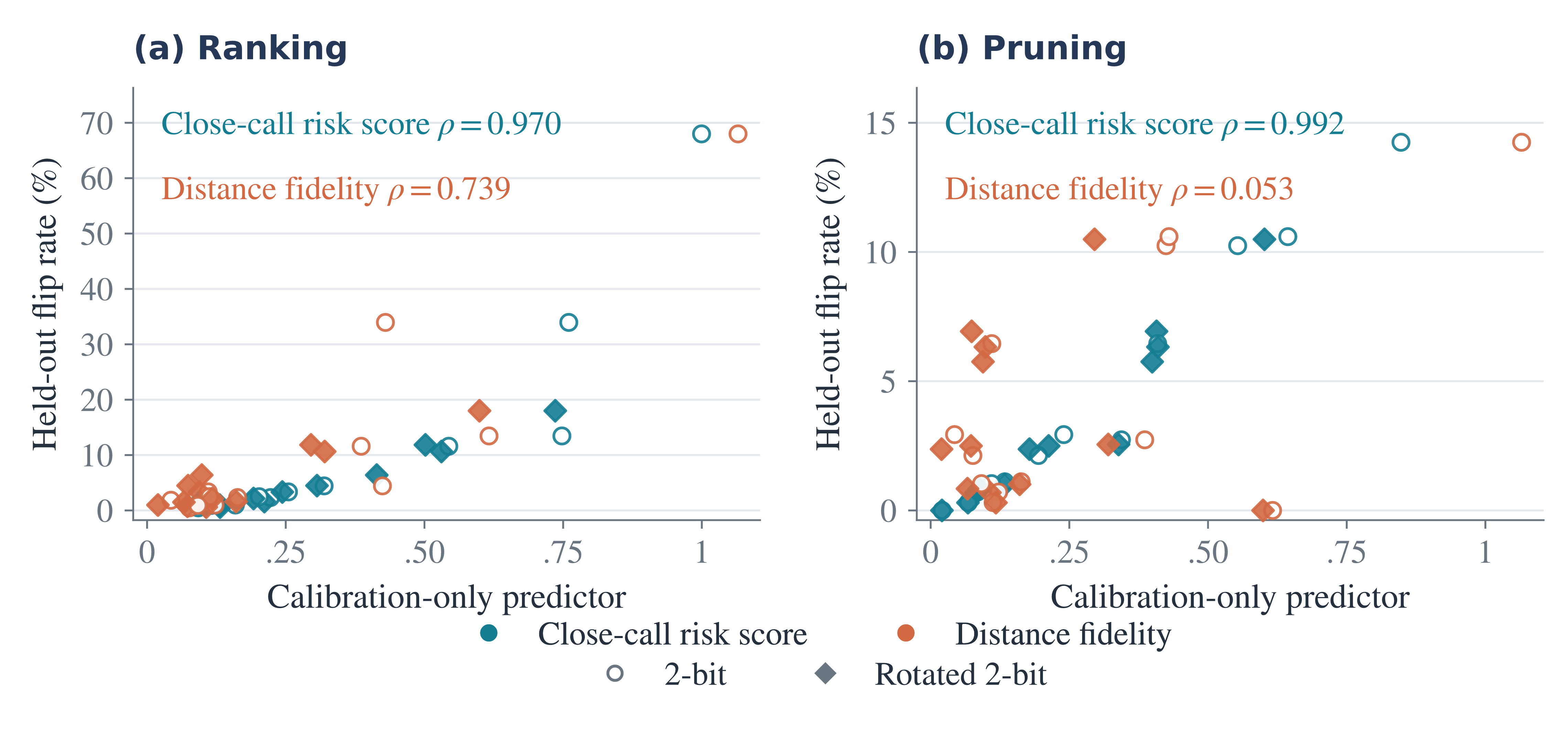}
\small
\caption{Held-out flip rates against two calibration-only predictors for 24 dataset--code configurations per decision type. Colours distinguish the close-call risk score and distance fidelity, shown as $1-\rho_{\mathrm{distance}}$ so that both increase with predicted risk. Shapes identify code variants. Annotations give each predictor's Spearman correlation with held-out flip rates.}
\label{fig:predictor-scatter}
\end{figure}

\begin{table}[tbp]
\centering
\small
\begin{tabular}{lrrrr}
\toprule
& \multicolumn{2}{c}{Across 24 configurations} & \multicolumn{2}{c}{Within-dataset agreement} \\
Ordering statistic & Ranking $\rho$ & Pruning $\rho$ & Ranking & Pruning \\
\midrule
Close-call risk score & .970 & .992 & 9/12 & 11/11 \\
Margin quantile, $-q_{.2}$ & .659 & .928 & tied & tied \\
Error scale, $v$ & $-.097$ & $-.177$ & 8/12 & 10/11 \\
Margin-to-noise summary, $-q_{.2}/v$ & .912 & .991 & 8/12 & 10/11 \\
Distance fidelity, $1-\rho_{\mathrm{distance}}$ & .739 & .053 & 7/12 & 9/11 \\
\bottomrule
\end{tabular}
\caption{Calibration-only ordering statistics and held-out flip rates. The first two columns correlate 24 configurations. The last two count agreement with the original-to-rotated change within each dataset, excluding one tied sphere-pruning pair. Here $q_{.2}$ is the 20th percentile of exact margins and $v$ the error scale. The margin quantile is identical for the two codes.}
\label{tab:predictor-controls}
\end{table}

Table~\ref{tab:predictor-controls} separates the two ingredients. The margin quantile alone predicts moderately, the error scale alone does not predict, and their ratio, a margin-to-noise summary, captures most of the predictive association. The paired cluster intervals for full score minus MNR are $[-0.075,0.318]$ for ranking and $[-0.020,0.056]$ for pruning. Estimating the two terms on separate pools and predicting the hardest fifth of decisions gives 0.955 for ranking and 0.983 for pruning; there the large-error term alone gives 0.953 for ranking and the close-call term alone 0.969 for pruning, so the two terms carry the prediction for different decision types (Appendix~\ref{app:predictor}).

\noindent\emph{Within-dataset code choice.}\quad
A shared-pool test compares the six families of Section~\ref{sec:instantiations} on Cohere, MiniLM, SIFT, and GIST, with 1{,}200 code-fitting, 800 calibration, and 2{,}000 evaluation vectors per dataset. Codes share exact 48-neighbour pools. On Cohere ranking, the comparison-fitted opposing score gives $\rho=1.000$, against $.314$ for MNR, $.714$ for global distance fidelity, $.143$ for same-pool fidelity, and $.886$ for same-pool distance MSE. On MiniLM, local fidelity gives $1.000$ and the comparison score $.943$. The opposing comparison score gives $.943$ for pruning on all four datasets. Table~\ref{tab:app-six-code-choice} reports every fixed variant, separating within-embedding code choice from cross-dataset association.

\subsection{Where Correlated Errors Come from (E2)}
\label{sec:exp-covariance}

On Cohere and GIST, Table~\ref{tab:covariance} reports error correlations well above the cap of $1/2$ that holds for i.i.d.\ endpoints and a fixed symmetric kernel. To locate the source, we keep one distance-fitted calibration fixed and vary only how comparisons are generated (Table~\ref{tab:selection-covariance}).

\begin{table}[tbp]
\centering

\setlength{\tabcolsep}{4pt}
\begin{tabular}{llrrrr}
\toprule
Dataset & Code & Independent & Random cand. & Top-256 & Top-32 / nearest-vs-rest \\
\midrule
Cohere & 1-bit & .361 & .312 & .524 & .525 / .544 \\
Cohere & 2-bit & .391 & .346 & .571 & .599 / .613 \\
Cohere & rotated 2-bit & .259 & .212 & .324 & .340 / .364 \\
MiniLM & 2-bit & .010 & .020 & .067 & .111 / .118 \\
GIST & 1-bit & .386 & .370 & .925 & .934 / .870 \\
GIST & 2-bit & .405 & .422 & .576 & .640 / .616 \\
\bottomrule
\end{tabular}

\small
\caption{Correlation between the two distance errors of a ranking comparison under different ways of generating comparisons, with one fixed distance-fitted calibration.}
\label{tab:selection-covariance}
\end{table}

Sampling endpoints with or without repetition changes the correlation by at most 0.003, while moving from random candidates to the top 32 raises it by a factor of 1.6 to 1.7 on Cohere and up to 2.5 on GIST, and from 0.02 to 0.11 on MiniLM. The covariance splits exactly into a part within queries and a part between queries (Appendix~\ref{app:within-between}), and the split holds to floating-point precision in every regime. For nearest-vs-rest comparisons, the nearest neighbour's error is constant within a query, so the covariance comes entirely from differences between queries: quantization shifts the distances of different queries by different amounts, and each shift cancels within that query's comparisons.

The size of the effect depends on the calibration. On the sample of Table~\ref{tab:covariance}, replacing the comparison-fitted calibration by the distance-fitted one lowers the correlation from 0.870 to 0.584 and the cancellation factor from 7.42 to 2.36, and the larger sample of Table~\ref{tab:selection-covariance} gives 0.613 and 2.59, so the cancellation is present under both calibrations (Appendix~\ref{app:covariance}).

\subsection{First Divergence in Practice (E3)}
\label{sec:exp-trace}

All 960 sampled selection calls (three datasets, two codes, 160 targets each) satisfy Theorem~\ref{thm:trace} exactly: the quantized output differs from the exact one when some action differs, and only then. Table~\ref{tab:trace} shows what this means for output agreement.

\begin{table}[tbp]
\centering

\begin{tabular}{lccc}
\toprule
Dataset & Pruning flip rate & Jaccard, fixed order & Jaccard, re-sorted \\
\midrule
Cohere-768 & 13.06\% & 0.322 & 0.395 \\
MiniLM-384 & 1.92\% & 0.654 & 0.529 \\
GIST-960 & 7.91\% & 0.458 & 0.260 \\
\bottomrule
\end{tabular}

\small
\caption{RobustPrune without backfill on exact candidate pools: two-bit codes, $\alpha=1.2$, $\Delta=32$, scale-only calibration. ``Fixed order'' reuses the exact candidate order; ``re-sorted'' also orders candidates by quantized distance.}
\label{tab:trace}
\end{table}

Output lists differ far more often than single comparisons flip: some action differs in 99.4 to 100\% of calls, although only 1.9 to 13.1\% of pruning comparisons flip. An action is an OR over the whole selected list, and once a candidate is appended in only one of the two calls, the outputs stay different. With the scale-only calibration, every flipped pruning comparison is a missed prune, and rescaling by a positive constant leaves every action unchanged. Adding the intercept of an independently fitted distance calibration, with everything else fixed, produces spurious prunes alongside missed prunes in all six dataset--code configurations and changes the pruning flip rate from 13.06\% to 8.82\% on Cohere, from 1.92\% to 1.11\% on MiniLM, and from 7.91\% to 8.69\% on GIST, so how calibration changes pruning depends on the representation (Appendix~\ref{app:trace}).

Adding backfill after standard RobustPrune on the 64-candidate pools ($\Delta=32$) fills 34.1\% of final edges on Cohere, none on MiniLM, and 28.8\% on GIST. The paired replay shares the exact candidate order during selection and backfill, so the plain-call bound also bounds backfilled differences. Re-sorted backfill additionally incurs the order-change term of Appendix~\ref{app:p4}.

\subsection{Held-Out Certificates (E4)}
\label{sec:exp-certificates}

Across 24 configurations per decision type, certificates from 512 blocks average 16.74\% / 18.65\% for the mechanism route and 9.07\% / 10.45\% for the direct failure route (ranking / pruning). At 1{,}024 blocks they fall to 12.74\% / 14.65\% and 6.75\% / 7.99\%. Every rate on 1{,}024 independent validation blocks lies below its certificate. Appendix~\ref{app:heldout} reports acceptance rates and the full sample-size comparison.

The direct certificate bounds flips, while the mechanism certificate bounds the close-call-or-large-error event. Each selects its own acceptance rule. Fixing an MNR cutoff on fitting data to retain the highest-margin fifth keeps 19.9\% / 20.9\% of ranking / pruning validation decisions and lowers flips from 8.58\% / 3.84\% to 5.27\% / 0.12\%.

For score-observable acceptance, six configurations use Cohere, MiniLM, and GIST two-bit codes on original and rotated axes. Each fixed 4{,}000-vector empirical pool supplies independently redrawn blocks with 32 distinct candidates. Separate random streams generate 128 fitting, 512 certification, and 512 validation blocks. The rule accepts when $|\widehat\Psi_Q|$ exceeds a cutoff. Five retention cutoffs and seven error thresholds are frozen on fitting blocks. Table~\ref{tab:observable-certificate} compares exact- and quantized-margin rules, each selected with a finite-family certificate at $\delta=0.05$. All validation rates lie below their certificates.

\begin{table}[!htbp]
\centering
\small
\setlength{\tabcolsep}{4pt}
\begin{tabular}{llrrr}
\toprule
Acceptance & Certificate & Value & Acceptance rate & Failure rate \\
\midrule
Exact margin & failure & 9.92\% / 12.51\% & 86.86\% / 72.87\% & 2.14\% / 3.05\% \\
Exact margin & mechanism & 20.52\% / 23.91\% & 65.46\% / 59.54\% & 0.70\% / 1.49\% \\
Quantized margin & failure & 9.50\% / 12.43\% & 83.53\% / 72.26\% & 1.28\% / 2.82\% \\
Quantized margin & mechanism & 25.30\% / 33.03\% & 58.91\% / 49.30\% & 0.22\% / 1.91\% \\
\bottomrule
\end{tabular}
\caption{Exact- and quantized-margin acceptance, averaged over six dataset--code configurations. Entries are ranking / pruning. The selected rule is evaluated on 512 independent validation blocks from the fixed empirical pool.}
\label{tab:observable-certificate}
\end{table}

\subsection{Resolving the Rotation Paradox (E5)}
\label{sec:exp-rotation}

A random rotation preserves every exact distance, so it leaves the decisions of our protocols and their exact margins unchanged, and with them the close-call term of Theorem~\ref{thm:main}. Any change in flip rates must therefore come from the comparison error, and Figure~\ref{fig:rotation} follows it for Cohere's codes.

For nearest-vs-rest comparisons on the same 160 held-out queries, distance rank correlation improves from 0.747 to 0.907 (panel a), while two-bit flips rise from 5.83\% to 8.81\%. These changes are invariant to positive calibration slopes. The residual variance decomposition depends on the fitted map. On the original 80-query fitting set, comparison-fitted calibration gives a 23\% rise in comparison-error variance. On the held-out set, that rise is 2\%, and distance-fitted calibration gives a 60\% decrease (panel b). The held-out standardized residual bias $\mu_R/v$ in exact-distance units ($c_Q=1$) changes from $-0.29$ to $-0.52$ under comparison fitting and from $+0.66$ to $-0.32$ under distance fitting. The comparison error relative to the signed margin governs the flip event, and its variance alone does not specify that event. Appendix~\ref{app:covariance} reports the full calibration cross, including signal scales, biases, and effective noise.

\begin{figure}[tbp]
\centering
\includegraphics[width=\textwidth]{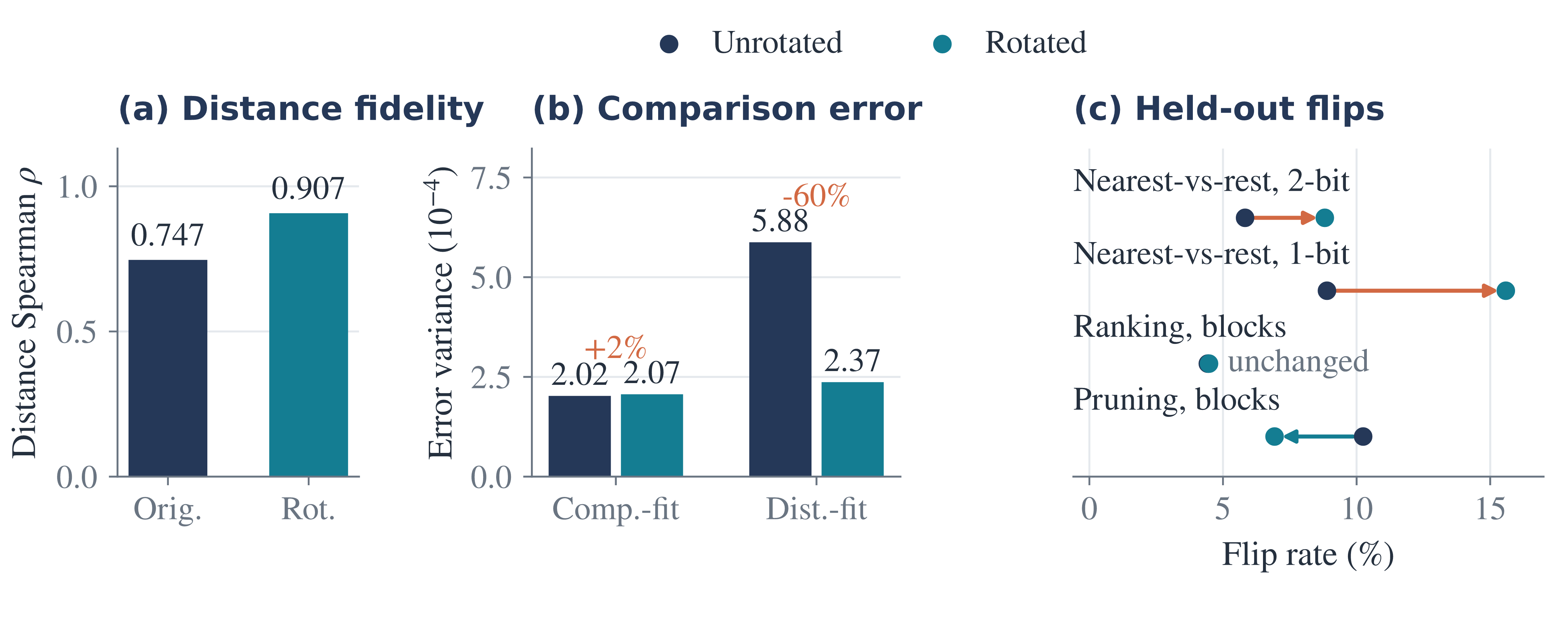}
\small
\caption{The rotation paradox on Cohere. Panels (a), (b), and the first two rows of (c) use the same 160 held-out queries and nearest-vs-rest sets. (a) Distance rank correlation for two-bit codes. (b) Comparison-error variance in exact-distance units under comparison-fitted and distance-fitted calibration. (c) Calibration-invariant held-out ranking flips, followed by two-bit ranking and pruning flips from the separate 48-candidate-block protocol of Section~\ref{sec:exp-margins}.}
\label{fig:rotation}
\end{figure}

The 48-candidate blocks show the same mechanism with different outcomes. There rotation raises the distance rank correlation of the calibration pairs from 0.576 to 0.926, and the scale of the comparison error falls by about a fifth for both ranking and pruning. For pruning, the bias of the comparison error also nearly vanishes, the held-out large-error probability at the threshold fixed on calibration data falls from 20.3\% to 6.6\%, and the flip rate falls from 10.25\% to 6.93\%. For ranking, the comparison error acquires a bias of two thirds of its new scale, the large-error probability stays at 13.2 to 13.3\%, and the flip rate stays at 4.4 to 4.5\%. In all three decision sets, the joint comparison error and margin explain the flip event. The close-call risk score of Section~\ref{sec:exp-margins}, which is built from comparison errors, orders the rotated and unrotated codes of a dataset in the same way as their held-out flip rates on 9 of 12 datasets for ranking and on all 11 datasets without tied flip rates for pruning (Table~\ref{tab:predictor-controls}).

\subsection{When the Gaussian Model Applies (E6)}
\label{sec:exp-eligibility}

The directly recomputed coordinate-wise QQ fits separate the archived learned representations from GIST. Joint checks examine an additional property. Testing 24 random subsets of 8 coordinates per dataset with a Kolmogorov--Smirnov test on the Mahalanobis radius at the 1\% level, BGE-M3 and GloVe reject none, Landmark-DINO rejects 8\%, MiniLM 13\%, Cohere 29\%, and Wolt-CLIP 58\%, while SIFT and GIST reject every subset.

Where the model applies, its predictions are accurate. The variance formula~\eqref{eq:oracle-variance} reproduces the measured ranking and pruning error variances within 4.6\% across twelve configurations, with a median error of 1.2\%. Over 28 configurations, the Efron--Stein bound exceeds the measured variance by a median factor of 1.67 (range 1.21 to 2.04), consistent with the factor of at most two, whereas a componentwise Cauchy--Schwarz bound exceeds it by a median factor of 40.8. The difference comes from the two components of the error, the bounded-code part and the bilinear part: their covariance is negative in all 28 configurations, with a median cancellation ratio of 0.947.

GIST shows the checks diagnosing a failure and its repair. Unrotated one-bit GIST has a calibration slope near zero, no correlation between quantized and exact margins, error kurtosis 26, a left tail 7.7 times the Gaussian value, and a threshold error equal to the whole error variance; the two-bit left tail is 7.8 times the Gaussian value. These diagnostics require no search run. After a random rotation, the margin correlation rises to 0.510 and 0.669 for one- and two-bit codes, the kurtosis falls to 3.5 and 2.8, the left tail to 1.88 and 0.76 times the Gaussian value, and the threshold error to 2.3\% of the variance. The flip rates fall from 100\% to 14.80\% for one-bit codes and from 35.83\% to 9.33\% for two-bit codes (Table~\ref{tab:covariance}).

The parameters of the two analytical tail routes can also be measured (Appendix~\ref{app:eligibility}). Across 24 configurations (six datasets, two codes, two decision types), the scale parameter $b_\star$ of the conditional tail route lies between 0.133 and 0.324 error standard deviations, clipping 1\% of the probability mass lowers its median from 0.158 to 0.147, every configuration admits a clipping bound below one, and no held-out tail exceeded its bound.

\subsection{Why Both Terms Are Needed (E7)}
\label{sec:exp-necessity}

Two constructions show that neither term of the close-call bound can be dropped. At an exact margin of zero, errors of arbitrarily small size select opposite actions, so the close-call term is necessary. Conversely, two perturbations of the same exact scores, one with light Gaussian tails and one concentrated near the decision boundary, have identical rank correlations with the exact scores (0.9860281863), yet they flip 5.09\% against 10.00\% of all decisions and 23.91\% against 50.00\% of the hardest fifth. Mean squared error orders them the wrong way, 0.0251 for the light-tailed perturbation and 0.0170 for the boundary-concentrated one, so choosing by either distance-level summary picks the code with twice the failure rate. The margin-to-noise ratio separates the two at once.

\subsection{Graph Construction and Search with QuIVer (E8)}
\label{sec:exp-graph}

Eight QuIVer graphs pair 200{,}000-vector Cohere and Wolt collections, insertion orders A/B (seeds 20261110/20261111), and original/signed-DCT-rotated axes. The builder uses $m=32$, 64 layer-zero neighbours, construction width 128, $\alpha=1.2$, and four threads. The first 128 queries fit diagnostics and the remaining 616 evaluate quality. Widths 64/128/256 and FP32 budget 64 are prespecified (Table~\ref{tab:actual-graph}). One Xeon Platinum 8352V CPU times five warmed-up complete native rounds. Every timed label matches the instrumented run.

Native expansion uses the frozen path's cheap two-bit Hamming score, followed by full sign--magnitude (SM) candidate selection and FP32 reranking. Same-graph controls use full SM, FP32, or the other axis's full SM for navigation and candidate selection. These builds have one layer. All scorers share its entry and FP32 budget.

At width 128, Cohere native recall falls from 95.69\% to 88.70\% under rotation, with work falling from 3{,}920 to 3{,}189 scored nodes. Full SM on the original graph gives 95.34\%, and rotated codes on that graph give 91.82\%. Original codes on the rotated graph recover 92.24\%. FP32 gives 98.69\% on the original graph and 93.71\% on the rotated graph, separating scoring and topology effects. Sampled neighbour Jaccard is 0.262--0.263, and reachability is 99.94\% versus 96.5\%. The full sweep also separates their recall--work frontiers. Original-axis Cohere at width 64 gives 93.29\% recall, 2{,}386 scored nodes, and 2{,}726 QPS, exceeding rotated width 256 at 89.96\%, 5{,}194 nodes, and 1{,}250 QPS. Wolt B at width 256 gives 90.70\% versus 87.31\% recall at nearly equal work (3{,}867 versus 3{,}851 nodes).

In each graph's width-128 FP32-search visited set, the exact nearest candidate is compared with the next 31. Averaged over both insertion orders, Cohere full-SM flips rise from 3.31\% to 5.78\% and its comparison-fitted score from 0.306 to 0.389. Cheap-score flips rise from 3.43\% to 4.73\%, while its score changes from 0.337 to 0.326. Comparison fitting retains this cheap-score inversion with opposing errors (0.292 to 0.283). Local distance fitting with opposing errors gives 0.194 to 0.267. Appendix~\ref{app:calibration-direction} gives the fitting rules and sampled comparisons.

Wolt A's entry lies in a closed directed set containing 63 and 5 copies of two vectors with cosine 0.999949. Three instrumented rebuilds reproduce the closure. In the original-axis single-thread run, the entry's final prune has 104 candidates, including 37 outside this set. The first selected point duplicates the target, giving raw SM scores $h(c,s)=h(c,t)$. The native test $h(c,s)<1.2h(c,t)$ removes the remaining positive-score candidates. One point survives, then nearest-score backfill fills the other 63 slots from this set, leaving zero exits. The four-thread run removes 36 external candidates.

Changing the entry to the first internal node outside the set (ID 1) gives 88.75\% / 85.16\% native recall (original / rotated, width 128). One entry edge to node 1 gives 88.02\% / 85.23\%. Fixed-pool entry self-scores $h(t,t)$ are 948 / 858. The alternative control $\alpha h(c,s)\le h(c,t)$ also retains zero exits on raw scores. Reconstructed-code squared $L_2$ with standard $\alpha d(c,s)\le d(c,t)$ retains 11--13 external neighbours. Applying that selection only to the entry's 64-neighbour list gives 75.78\% / 82.27\% recall in the single-thread rebuilds. Exact-cosine entry selection gives 76.66\% on original axes.

% ============================================================
\section{Practical Guidance}
\label{sec:guidance}

The results suggest a decision-level audit for comparing low-bit codes on an index workload.

\begin{enumerate}
\setlength{\itemsep}{2pt}
\setlength{\parsep}{0pt}
\item \textbf{Sample workload decisions.} Build one code-independent reference pool using exact top-$M$ neighbours on a small sample or one FP32 pilot graph. Reuse its query candidates and pruning lists for all candidate codes, with separate fitting and evaluation blocks.
\item \textbf{Measure margins.} Estimate close-call frequencies and the near-tie exponent independently of the candidate code.
\item \textbf{Calibrate and inspect errors.} Fit each code and apply the heuristic diagnostic screen. Use held-out certificates when analytical assumptions are unavailable.
\item \textbf{Compare decision quantities.} Evaluate the full risk score, MNR, and held-out flips alongside actual graph recall and work.
\item \textbf{Certify acceptance rules.} Report each accepted-decision certificate together with coverage and validation failure, using quantized margins for score-observable rules.
\end{enumerate}

\noindent\emph{Assessing rotation.}\quad
Inspect margin association, residual bias, tail frequencies, and held-out flips jointly. The applicability screen summarizes these diagnostics, with cutoff sensitivity reported in Appendix~\ref{app:eligibility}. The entropy gap $1-H_{\mathrm{sign}}$ tracks rotation's change in distance fidelity. Decision quantities and graph recall--work curves evaluate its effect on search. Cohere and SIFT have nearly equal sign entropy (0.747 and 0.746), but different error diagnostics and responses to rotation.

% ============================================================
\section{Related Work}
\label{sec:related}

\noindent\emph{Binary codes and random projections.}\quad
SimHash~\citep{charikar2002similarity} shows that random-hyperplane sign bits preserve angular similarity, with collision probability $1-\arccos\langle x,y\rangle/\pi$ for fixed vectors. Its randomness lies in the projection directions, so it describes random hyperplanes rather than the fixed coordinate axes of a learned encoder, whose coordinate variances can differ by an order of magnitude. RaBitQ~\citep{gao2023rabitq} gives a sharp pointwise error bound for a randomized estimate of a single inner product. Product quantization~\citep{jegou2011product,ge2014optimized} and ScaNN~\citep{guo2020accelerating} use learned codebooks at 4 to 64 bits. These results describe individual distance estimates; we study differences of estimates that share endpoints.

\noindent\emph{Graph-based nearest-neighbour search.}\quad
HNSW~\citep{malkov2020efficient} and Vamana~\citep{subramanya2019diskann} build navigable graphs by repeated candidate sorting and diversity pruning, and QuIVer builds and searches its graph with two-bit scores~\citep{xiao2026quiver}. The QuIVer paper designs and evaluates that index; this paper develops the decision-level analysis of when its codes, and low-bit codes in general, preserve search decisions. Existing theory of graph search studies navigability under random-graph or geometric assumptions. We ask whether quantization preserves the local decisions made during construction and search, and answer it with an exact theorem about the selection procedure.

\noindent\emph{Contrastive representation geometry.}\quad
\citet{betser2026infonce} prove asymptotic Gaussianity of fixed-dimensional projections of InfoNCE-trained representations under stated alignment and concentration assumptions, with a second regime that adds a vanishing regularizer. We use this result to motivate the exact Gaussian model of Section~\ref{sec:bridge}, and Theorem~\ref{thm:counterexample} gives a rare-spike obstruction for the population-radius idealization under low-order Gaussian checks.

\noindent\emph{Concentration and margin conditions.}\quad
The close-call bound transfers the low-noise condition of \citet{tsybakov2004optimal} from classification to the comparisons of a search algorithm. The large-error term draws on sub-Gaussian and sub-gamma tail bounds~\citep{vershynin2018high,boucheron2013concentration}, Efron--Stein bounds~\citep{efron1981jackknife}, and empirical Bernstein inequalities~\citep{maurer2009empirical}. Our framework combines errors specific to each decision type, covariance-aware error scales, and exact-run replay of graph construction.

% ============================================================
\section{Conclusion}
\label{sec:conclusion}

A graph index uses a low-bit code through comparisons. The close-call bound separates small margins from large comparison errors. Shared-endpoint covariance determines difference errors, and exact-run replay connects candidate actions to neighbour-list changes. The Gaussian model supplies reference errors, and held-out certificates cover exact- and quantized-margin rules. Across 24 dataset--code configurations, risk and MNR correlate more closely with held-out flips than distance fidelity. Paired graphs connect these local effects to retrieval quality and search work.

\noindent\emph{Open directions.}\quad
Open questions concern predicting useful path-bound lengths from representation geometry, extending first divergence to adaptive beam-search candidates, and choosing rotation, bit width, and verification cutoffs from pilot covariance statistics.

\acks{No external funding was received for this work. The authors declare no competing interests.}

\FloatBarrier
\appendix
\section{Proofs for Section~\ref{sec:flip}}
\label{app:p1}

The close-call bound uses a sign argument at each comparison, while its optimized rate balances the close-call probability against the error tail.

\subsection{Ties and the Flip Event}

Write $\widehat\Psi_Q=c_Q\Psi+R$ with $c_Q>0$. If $\Psi>0$ and a flip occurs, then $\widehat\Psi_Q\le0$ and
$R\le-c_Q\Psi=-c_Q\Gamma$.
If $\Psi<0$, a flip gives $R\ge c_Q\Gamma$. Hence
\begin{equation}
\{\text{flip}\}\subseteq\{\Gamma=0\}\cup\{\Gamma>0,\ |R|\ge c_Q\Gamma\},
\label{eq:app-p1-inclusion}
\end{equation}
and multiplying by $-\sign(\Psi)$ gives $R_{\mathrm{opp}}\ge c_Q\Gamma$. Splitting at $\tau$ proves both parts of Theorem~\ref{thm:main}. Exact ties lie in the close-call term. Quantized ties are covered by the non-strict error event under every fixed tie rule.

\subsection{Explicit Rate}

Under Corollary~\ref{cor:subgaussian} write $r=\nu/c_Q$ and $A=Cr^\beta$. When $0<A<2$, let $L=\log(2/A)$ and $\tau_r=r\sqrt{2L}$. If $\tau_r\le\tau_0$, substituting $\tau_r$ into~\eqref{eq:small-ball-bound} gives
\begin{equation}
\Pp(\text{flip})\le Cr^\beta(2L)^{\beta/2}+Cr^\beta,
\label{eq:app-p1-rate}
\end{equation}
so the worst case over the class of distributions allowed by the two assumptions is
$O\!\left(r^\beta[\log(1/r)]^{\beta/2}\right)$.
The logarithm reflects the worst case over this class. When the error tail is controlled conditionally on the margin, direct integration gives $O(r^\beta)$. Corollary~\ref{cor:subgaussian} assumes zero mass at exact ties. More generally, if $\Pp(\Gamma\le\tau)\le p_0+C\tau^\beta$ with $p_0=\Pp(\Gamma=0)$, the same argument adds $p_0$ to the rate.

\subsection{Ranking Error and Bias}
\label{app:ranking-proof}

Let $X=(\xi_x,\xi_y)^\top$, $m=\E X$, and $u=(-1,1)^\top$. The joint MGF proxy gives
$\E\exp\{\lambda u^\top(X-m)\}\le\exp\{\lambda^2u^\top Ku/2\}$,
so $R_{\mathrm{rank}}-\mu_R$ is sub-Gaussian with
$\mu_R=u^\top m,\qquad \nu_{\mathrm{rank}}^2=u^\top Ku=K_{xx}+K_{yy}-2K_{xy}$.
For a fixed positive exact margin $\gamma$ the flip event is $R_{\mathrm{rank}}\le-c_Q\gamma$, and Chernoff's method gives Proposition~\ref{prop:ranking}. For the opposite orientation the relevant right tail has effective margin $c_Q\gamma-\mu_R$. If $\nu_{\mathrm{rank}}=0$, Jensen's inequality and the MGF bound give $R_{\mathrm{rank}}=\mu_R$ almost surely, and the decision is deterministic.

\subsection{Shared-Endpoint Covariance}
\label{app:p2p3}

\noindent\emph{Independent endpoints.}\quad
Let $X,Y,Z$ be i.i.d.\ and let $h\in L^2(P\otimes P)$ be symmetric, with Hoeffding decomposition
$h(x,y)-\mu=h_1(x)+h_1(y)+h_2(x,y)$,
where $h_2$ is degenerate in each argument. Put $a=\Var(h_1(X))$ and $b=\Var(h_2(X,Y))$. Orthogonality gives
$\Var(h(X,Y))=2a+b,\qquad\Cov(h(X,Y),h(X,Z))=a$,
so that, with $R_1=h(X,Y)$, $R_2=h(X,Z)$, and positive marginal variance,
\begin{equation}
0\le\Corr(R_1,R_2)=\frac{a}{2a+b}\le\frac12,
\qquad
1\le\frac{\Var(R_1)+\Var(R_2)}{\Var(R_1-R_2)}=\frac{2a+b}{a+b}\le2.
\label{eq:app-iid-half}
\end{equation}
The lower bounds are attained at $a=0,b>0$ and the upper bounds at $b=0,a>0$; if $a=b=0$, both ratios are undefined. The second ratio is the cancellation factor of Section~\ref{sec:ranking}.

\noindent\emph{Search-selected comparisons.}\quad
Let $R_L$ and $R_R$ be the two distance errors of a comparison with query $Q$, and write $m_L(Q)=\E[R_L\mid Q]$, $v_L(Q)=\Var(R_L\mid Q)$, analogous quantities for $R_R$, and $c(Q)=\Cov(R_L,R_R\mid Q)$. The laws of total covariance and total variance give
\begin{align}
\Cov(R_L,R_R)&=\E c(Q)+\Cov(m_L(Q),m_R(Q)),\\
\Var(R_L-R_R)&=\E[v_L(Q)+v_R(Q)-2c(Q)]+\Var(m_L(Q)-m_R(Q)).
\label{eq:app-selected-covariance}
\end{align}
When the candidates of a query are conditionally independent draws from a fixed selection kernel, $c(Q)=0$, and the pooled correlation comes entirely from the second term; it approaches one when query-specific means dominate the conditional variances. Selection changes the law of the candidates, and with it the conditional means and variances; it changes the weights of the queries only if query sampling, retention, or row weighting also changes. For uniform sampling without replacement from a pool of size $N_{\mathrm{pool}}$,
\begin{equation}
\Cov(f_I,g_J\mid Q)=-\frac{c_{fg}(Q)}{N_{\mathrm{pool}}-1}.
\end{equation}
The error of a fixed nearest neighbour is constant given the query, so it has zero conditional covariance, and its within-query correlation is undefined. These identities concern second moments; the MGF proxy $K$ of Proposition~\ref{prop:ranking} is a separate input.

\subsection{Within- and Between-Query Identity}
\label{app:within-between}

For query group $g$ with $n_g$ observed pairs $(\ell_{gi},r_{gi})$, let $N=\sum_gn_g$, let $\bar\ell_g,\bar r_g$ be the group means, and let $\bar\ell,\bar r$ be the row-weighted grand means. Expanding each centred product gives the deterministic identity
\begin{equation}
\sum_{g,i}(\ell_{gi}-\bar\ell)(r_{gi}-\bar r)
=\sum_{g,i}(\ell_{gi}-\bar\ell_g)(r_{gi}-\bar r_g)
+\sum_gn_g(\bar\ell_g-\bar\ell)(\bar r_g-\bar r),
\label{eq:app-anova-covariance}
\end{equation}
so that, with the pooled denominator $N-1$,
\begin{equation}
s_{\mathrm{pool}}=\sum_{g:n_g\ge2}\frac{n_g-1}{N-1}s_g+\frac1{N-1}\sum_gn_g(\bar\ell_g-\bar\ell)(\bar r_g-\bar r).
\end{equation}
The same identity decomposes both marginal variances and therefore $\Var(R_L-R_R)$. It holds for any dependence between rows; reading its two parts as unbiased estimators of the population terms in~\eqref{eq:app-selected-covariance} requires a sampling model.

\subsection{Top-$K$ Selection}

Let $C$ be fixed before the randomness of the code, let $A$ be the exact top-$K$ subset, and suppose that no exact tie crosses the boundary of $A$. The quantized top-$K$ set equals $A$ if every $x\in A$ stays ahead of every $y\in C\setminus A$, so
\begin{equation}
\Pq(\widehat A\ne A)\le\sum_{x\in A}\sum_{y\notin A}\Pq\{\widehat d_Q(q,y)\le\widehat d_Q(q,x)\},
\label{eq:app-topk}
\end{equation}
and any fixed set of comparisons that determines the top-$K$ set can replace the full cross product.

\subsection{Pruning}
\label{app:pruning}

RobustPrune drops candidate $c$ when a selected neighbour $s$ satisfies $\alpha d(s,c)\le d(t,c)$ with $\alpha\ge1$. Define
\begin{equation}
\Psi_\alpha=d(t,c)-\alpha d(s,c),\qquad
\widehat\Psi_{\alpha,Q}=c_Q\Psi_\alpha+Z_\alpha,\qquad
Z_\alpha=\xi_{tc}-\alpha\xi_{sc},
\end{equation}
with action $\mathbf 1\{\Psi_\alpha\ge0\}$. For nonzero margins, the missed-prune and spurious-prune events are
\begin{equation}
\{\Psi_\alpha>0,\ \widehat\Psi_{\alpha,Q}<0\},\qquad
\{\Psi_\alpha<0,\ \widehat\Psi_{\alpha,Q}\ge0\},
\end{equation}
and at an exact tie the actions differ on $\{\Psi_\alpha=0,\ Z_\alpha<0\}$. A continuous distribution of $Z_\alpha$ rules out quantized ties but leaves the probability of this event unconstrained, so it enters the risk as a term of its own.

With $w=(1,-\alpha)^\top$, a joint MGF proxy $K^{(3)}$ for the errors of the two distances gives
$\nu_\alpha^2=w^\top K^{(3)}w=K_{tc,tc}+\alpha^2K_{sc,sc}-2\alpha K_{tc,sc}$.
Writing $\mu_Z=\E Z_\alpha$, the two orientations have effective one-sided margins $c_Q\gamma+\mu_Z$ and $c_Q\gamma-\mu_Z$. A calibration through the origin does not make $\mu_Z$ vanish, and a shared intercept $b$ in the distance calibration enters $Z_\alpha$ as $(1-\alpha)b$. If both distance errors are bounded by $\epsilon$ deterministically, then $|Z_\alpha|\le(1+\alpha)\epsilon$.

% ============================================================
\section{Proofs for Section~\ref{sec:trace}}
\label{app:p4}

Append-only selection preserves the first action mismatch because a distinct candidate appended in only one run cannot be removed later.

\subsection{Selection as a Sequence of Actions}

Fix distinct candidate labels in an order $\pi=(c_{(1)},\ldots,c_{(N)})$ sorted by exact distance with a fixed tie-breaking rule, a target $t$, a degree budget $\Delta$, and $\alpha\ge1$. Starting from $S_1=()$, define at each visited position
\begin{equation}
a_i=\bigvee_{s\in S_i}\mathbf 1\{d(t,c_{(i)})-\alpha d(s,c_{(i)})\ge0\}.
\end{equation}
Append $c_{(i)}$ exactly when $a_i=0$, and stop when the order is exhausted or the degree budget is reached. The call only appends and has no backfill. For a fixed sorted order it is equivalent to repeatedly selecting the nearest remaining candidate and deleting the candidates it blocks.

For each state visited by the exact call, evaluate the quantized action $\hat a_i^*$ on the exact list $S_i$ and write $E_i=\{\hat a_i^*\ne a_i\}$. Induction on the first mismatch gives the deterministic identity
\begin{equation}
\{\widehat T\ne T\}=\{\widehat S_{\mathrm{out}}\ne S_{\mathrm{out}}\}=\bigcup_{i\in\mathcal V}E_i.
\label{eq:app-trace-identity}
\end{equation}
The step from a difference in output back to some $E_i$ uses distinct labels, append-only selection, and the absence of backfill: the first candidate with an action mismatch belongs to exactly one of the two final lists.

\subsection{Blocking Neighbours and the Risk Sum}

For $s\in S_i$ let
\begin{equation}
A_{i,s}=\bigl\{\mathbf 1\{\Psi_{i,s}\ge0\}\ne\mathbf 1\{\widehat\Psi_{i,s}\ge0\}\bigr\}
\end{equation}
be the event that the pruning comparison flips. If the exact action keeps the candidate, $E_i=\bigcup_{s\in S_i}A_{i,s}$. If it prunes, the blocker set $B_i=\{s\in S_i:\Psi_{i,s}\ge0\}$ is nonempty and $E_i\subseteq\bigcap_{s\in B_i}A_{i,s}$. Thus $\Pp(E_i)$ is bounded by the sum in the first case and by the smallest blocker probability in the second. For $\widehat\Psi=c_Q\Psi+R$ and nonzero margins,
\begin{equation}
A\cap\{\Psi\ne0\}\subseteq\{\Psi\ne0,\ \Psi\widehat\Psi\le0\}\subseteq\{\Psi\ne0,\ |R|\ge c_Q|\Psi|\},
\end{equation}
where the first inclusion can be strict at a quantized tie. Exact ties flip when $R<0$. Valid pairwise bounds therefore come from
\begin{equation}
p_{i,s}=\begin{cases}
\Pp(-R_{i,s}\ge c_Q\Psi_{i,s}),&\Psi_{i,s}>0,\\
\Pp(R_{i,s}\ge c_Q|\Psi_{i,s}|),&\Psi_{i,s}<0,\\
\Pp(R_{i,s}<0),&\Psi_{i,s}=0.
\end{cases}
\label{eq:app-trace-risk}
\end{equation}
Applying~\eqref{eq:action-bound} and summing $q_i$ over exact-run states proves~\eqref{eq:trace-bound}, with no independence assumption. The coarser double sum follows from $q_i\le\sum_{s\in S_i}p_{i,s}$.

\subsection{Edges, Paths, Reordering, and Backfill}

If the exact output edge $(t,c)$ is selected at position $j_c$, agreement of all actions up to $j_c$ suffices for its survival. Bounds for the edges of a fixed path, each computed for its own source, combine by a union bound without independence across sources; the resulting statement concerns the retention of the listed edges.

If quantization also changes the candidate order, the event that the order changes is added:
\begin{equation}
\Pp(\text{the outputs differ})\le\Pp(\widehat\pi\ne\pi)+\sum_{i\in\mathcal V}\Pp(E_i).
\end{equation}
The same deterministic backfill map $B_\pi$, which appends unselected candidates in the shared order $\pi$, satisfies
\begin{equation}
\{\text{the backfilled outputs differ}\}\subseteq\{\text{the decision records without backfill differ}\},
\end{equation}
and the reverse inclusion fails in general, because different decision records can be backfilled to the same list. If the runs use different backfill orders, the order-change event must also be included as above.

% ============================================================
\section{Derivations for Section~\ref{sec:bridge}}
\label{app:p5}

The Gaussian model separates the idealized coordinate error from errors due to empirical thresholds and vector normalization.

\subsection{Decomposition of the Error}

Let $G_u\overset{\mathrm{i.i.d.}}\sim\mathcal N(\mu_D,\Sigma_D)$ and $X_u=G_u/\|G_u\|$. Define the per-vector threshold and the two-bit code
\begin{equation}
T(g)=D^{-1}\sum_i|g_i|,\qquad q_i(g)=\sign(g_i)[1+\mathbf 1\{|g_i|>T(g)\}].
\end{equation}
The code is invariant to positive scaling, so $q_i(X)=q_i(G)$. Replacing $T(G)$ by $\overline T_D=\E T(G)$ and $\|G\|$ by $r_D=\E\|G\|$ defines the idealized error, and for every fixed combination of node pairs
$R_a=R_a^\circ+R_{a,\mathrm{thr}}+R_{a,\mathrm{norm}}$,
where $a$ indexes the decision type. The threshold error is nonzero only when some coordinate falls on different sides of the vector's own threshold and the population threshold. The normalization error is the difference between exact cosine normalization and the bilinear form with the fixed radius $r_D$.

\subsection{Magnitude-Bit Gain}

Let $G,H$ be independent with independent coordinates $G_i,H_i\sim\mathcal N(0,\sigma_i^2)$. Put $m=\sqrt{2/\pi}$, $z_i=\tau/\sigma_i$, $p_i=2[1-\Phi(z_i)]$, $k_i=m+2\varphi(z_i)$, and $v_i=1+3p_i$. Independence across coordinates gives
\begin{align}
\Var(T)&=\sum_i\sigma_i^4,&\Cov(S_1,T)&=m^2\sum_i\sigma_i^2,&\Var(S_1)&=D,\\
\Cov(S_2,T)&=\sum_i\sigma_i^2k_i^2,&\Var(S_2)&=\sum_iv_i^2.
\end{align}
For every finite $z>0$,
\begin{equation}
k(z)^2>m^2v(z).
\label{eq:app-magnitude-scalar}
\end{equation}
To see this, write $a=e^{-z^2/2}$ and $h(z)=k(z)^2/m^2-v(z)=2a+a^2-6[1-\Phi(z)]$. Then $h(0)=\lim_{z\to\infty}h(z)=0$ and $h'(z)=a\{3m-2z(1+a)\}$; since $2z(1+a)$ is strictly increasing, $h$ first rises and then falls, and it stays strictly positive on $(0,\infty)$.

Now set $t_i=\sigma_i^2$ and regard $v_i=v(t_i)$. Both $v(t)$ and $t/v(t)$ increase, because $tv'(t)=3z\varphi(z)<1\le v(t)$. Pairwise expansion and the Cauchy--Schwarz inequality give
\begin{equation}
\frac{\sum_it_iv_i}{\sum_it_i}\ge\frac{\sum_iv_i^2}{\sum_iv_i}\ge\sqrt{\frac1D\sum_iv_i^2},
\end{equation}
and combining this with~\eqref{eq:app-magnitude-scalar} proves $\rho(S_2,T)>\rho(S_1,T)$.

The match between target and scorer is what makes the inequality hold. In the linear-target model $T=\sum_iw_iG_i$, define the single-vector readouts $S_1=\sum_i\sign(G_i)$ and $S_2=\sum_iq_i(G)$, and take $D=8$, $w=(1,1,1,1,0,0,0,0)$, coordinate standard deviations $(1,1,1,1,L,L,L,L)$, and $\tau=m(1+L)/2$. For $L=8$ the two-bit minus one-bit Pearson correlation is $-0.162687\ldots$: the second bit amplifies four coordinates that the target ignores, and a fixed readout loses. For independent equal-variance coordinates and the per-vector threshold $D^{-1}\sum_i|G_i|$, the strong law of large numbers and dominated convergence recover the fixed-threshold correlations as $D\to\infty$.

\subsection{Variance of the Idealized Error}

For the symmetric idealized kernel $H_D$ let
\begin{equation}
h_{1,D}(u)=\E[H_D(u,V)]-\theta_D,\qquad
h_{2,D}(u,v)=H_D(u,v)-\theta_D-h_{1,D}(u)-h_{1,D}(v),
\end{equation}
and $\kappa_j=\E h_j^2$. Aggregate directed or repeated pairs into coefficients $A_p$ on unordered pairs and net node coefficients $d_u=\sum_{p\ni u}A_p$. Hoeffding orthogonality gives
\begin{equation}
\Var(R_a^\circ)=\kappa_1\sum_ud_u^2+\kappa_2\sum_pA_p^2.
\label{eq:app-oracle-variance}
\end{equation}
Disjoint pairs have zero covariance, pairs that share exactly one node have covariance $\kappa_1$, and identical or reversed pairs have variance $2\kappa_1+\kappa_2$.

\subsection{Analytical Routes to the Tail}

A coarse route combines a Gaussian quadratic-form MGF with a bounded score range; it needs no independence between coordinates, and on real embeddings it is loose by four orders of magnitude (Appendix~\ref{app:eligibility}). The exact-covariance route keeps~\eqref{eq:app-oracle-variance} and adds a cumulant or conditional-MGF condition. The held-out route calibrates the error of each decision type on independent blocks. The three routes bound the same error tail with different inputs, and the close-call bound accepts any of them.

\subsection{The Rare-Spike Counterexample}
\label{app:p8}

Let $p_D=D^{-4}$, $a_D=D^{3/2}$, and
$G_D=(1-J_D)Z_D+J_Da_DS_D$,
where $J_D\sim\mathrm{Bernoulli}(p_D)$, $Z_D\sim\mathcal N(0,I_D)$, and $S_D$ has independent random-sign coordinates; independent nodes use independent copies. Then $\Cov(G_D)=(1+D^{-1}-D^{-4})I_D$. For every fixed block of $k$ coordinates, a mixture decomposition and a Gaussian scale comparison give total-variation distance $O_k(D^{-1})$ from the standardized Gaussian law. With $r_D=\E\|G_D\|\sim\sqrt D$,
\begin{equation}
\E\left(\frac{\|G_D\|}{r_D}-1\right)^2=O(D^{-1}),
\end{equation}
a relative mean-square statement whose $L^2$ norm is $O(D^{-1/2})$. Set $\tau_D=\E|G_{D1}|\sim\sqrt{2/\pi}$ and $\epsilon_D=D^{-1/4}$. The spikes eventually lie outside the band $[\tau_D-\epsilon_D,\tau_D+\epsilon_D]$. Since the density of $|Z_{D1}|$ is at most $2/\sqrt{2\pi}$, the expected fraction of coordinates with $\bigl||G_{Di}|-\tau_D\bigr|\le\epsilon_D$ is at most $4\epsilon_D/\sqrt{2\pi}=O(D^{-1/4})$.

Fix $0<\eta\le1/8$ and define the idealized ranking error with a fixed threshold,
\begin{align}
R_D^\circ&=Q_D+B_D,\\
Q_D&=\frac{G_1^\top(G_3-G_2)}{r_D^2},\\
B_D&=-\frac{\eta}{D}\sum_iq_i^\circ(G_1)\{q_i^\circ(G_3)-q_i^\circ(G_2)\},
\end{align}
where $q^\circ$ is the code with the population threshold. The coefficient $\eta/D$ has the same order $\Theta(D^{-1})$ as the population regression slope that relates the unscaled code sum to $Q_D$. Conditioning on the three spike indicators and combining the two pairs that share an endpoint coordinate by coordinate gives independent centred summands, whence
$\E B_D^2\le32\eta^2/D,\qquad\E B_D^4\le3072\eta^4/D^2$,
while $\Var(Q_D)=2D\sigma_D^4/r_D^4\sim2/D$. The reverse triangle inequality in $L^2$ gives $V_D=\Var(R_D^\circ)=\Theta(D^{-1})$.

On the spike event $\mathcal S_D=\{J_1=J_2=1,J_3=0\}$, whose probability is $p_D^2(1-p_D)=\Theta(D^{-8})$,
\begin{equation}
Q_D=\frac{a_DS_1^\top Z_3-a_D^2S_1^\top S_2}{r_D^2}.
\end{equation}
Since $\E(S_1^\top S_2)^4=3D^2-2D$,
\begin{equation}
\E[Q_D^4;\mathcal S_D]\ge\frac{p_D^2(1-p_D)a_D^8(3D^2-2D)}{r_D^8}=(3+o(1))D^2.
\end{equation}
The reverse triangle inequality in $L^4$ and the bound on $B_D$ give $\E(R_D^\circ)^4=\Omega(D^2)$, hence
$\kappa_4(R_D^\circ)=\Omega(D^2),\qquad\kappa_4(R_D^\circ)/V_D^2=\Omega(D^4)$.

For a cumulant condition $|\kappa_j(R_D^\circ)|\le(j!/2)V_Db_D^{j-2}$, the case $j=4$ forces $b_D/\sqrt{V_D}=\Omega(D^2)$. For a two-sided MGF bound
\begin{equation}
\log\E e^{\lambda X}\le\frac{\lambda^2v}{2(1-b|\lambda|)},\qquad|\lambda|<1/b,
\end{equation}
evaluating both signs at $\lambda=[2\max\{\sqrt v,b\}]^{-1}$ and using $\cosh y-1\ge y^4/24$ gives $\E X^4\le192v\max\{v,b^2\}$, so $v_D=O(V_D)$ again forces $b_D/\sqrt{V_D}=\Omega(D^2)$. The construction uses a fixed threshold, population radius, and idealized slope. Under per-vector unit normalization the spike amplitude cancels, so this construction establishes an obstruction for $R_D^\circ$. A conclusion for the cosine error $R$ also requires controlling $R_{\mathrm{norm}}$.

\subsection{Replacement Variance}

For the Hoeffding decomposition of the error $R_a^\circ$, independent replacement of node $u$ gives
\begin{equation}
\tfrac12\E\bigl(R_a^\circ-R_a^{\circ(u)}\bigr)^2=\E\Var(R_a^\circ\mid G_{-u}).
\end{equation}
Summing over nodes counts every first-order component once and every degenerate pair component twice,
\begin{equation}
V_{\mathrm{ES}}=\kappa_1\sum_ud_u^2+2\kappa_2\sum_pA_p^2,
\end{equation}
hence $\Var(R_a^\circ)\le V_{\mathrm{ES}}\le2\Var(R_a^\circ)$, with equality at two when $h_1=0$ and $h_2\ne0$.

\subsection{Conditional Tail and Clipping Routes}

Let $Y_1,Y_2,Y_3$ be the martingale increments of $R_a^\circ$ obtained by revealing the three nodes of a comparison in turn. If deterministic $v_j,b_j$ satisfy
\begin{equation}
\log\E[e^{\lambda Y_j}\mid\mathcal F_{j-1}]\le\frac{\lambda^2v_j}{2(1-b_j|\lambda|)},
\end{equation}
iterated conditioning gives the tail~\eqref{eq:conditional-mgf-tail} with $v_\star=\sum_jv_j$ and $b_\star=\max_jb_j$. Freedman's inequality needs these conditional bounds almost surely; an expected replacement proxy is not a substitute.

The clipping route constructs an error $R_L$ on the full input space with
$\Pp(R_L\ne R_a^\circ)\le\varepsilon_L,\qquad|\E R_a^\circ-\E R_L|\le d_L$.
If $R_L$ has a sub-gamma tail with parameters $(v_L,b_L)$, then
\begin{equation}
\Pp(|R_a^\circ-\E R_a^\circ|\ge z+d_L)\le2\exp\left[-\frac{z^2}{2(v_L+b_Lz)}\right]+\varepsilon_L.
\end{equation}
Bounding replacements only on a good event suffices when that event is stable under every single-node replacement, which the construction of $R_L$ on the full space provides.

\subsection{Negative Covariance of the Two Error Components}

For zero-mean Gaussian nodes of arbitrary covariance, define $M_{ij}=\E[q_i^\circ(G)G_j]$, and write the idealized error as the sum of a bounded-code part $B_a$ with coefficient $\eta_1$ and a bilinear part $Q_a$. Independence of the endpoints gives
\begin{equation}
\Cov(B_a,Q_a)=-\frac{\eta_1}{r_D^2}\Bigl(\sum_pA_p^2\Bigr)\|M\|_F^2\le0,
\end{equation}
without any symmetry or positive semidefiniteness of $M$. At nonzero mean, first-order projections reappear, and the sign of the covariance depends on separate cross-Hoeffding conditions.

% ============================================================
\section{Proofs for Section~\ref{sec:heldout}}
\label{app:p6}

Independent blocks allow the accepted failure rate to be bounded even when comparisons within a block are dependent.

\subsection{A Single Threshold}

Condition on the parameters fitted on the independent split. Let $Z_1,\ldots,Z_n$ be independent blocks, with arbitrary dependence among the decisions inside a block. For threshold $\tau$ let $u(Z,\tau)$ be the fraction of decisions in block $Z$ that are close calls or large errors, so that $0\le f(Z)\le u(Z,\tau)\le1$, with $f(Z)$ the fraction that flip. For a fixed $\tau$, one-sided Hoeffding gives
\begin{equation}
\E f(Z)\le\widehat U_n(\tau)+\sqrt{\frac{\log(1/\delta)}{2n}}
\end{equation}
with probability at least $1-\delta$. For a prespecified grid of $M$ thresholds, a union bound replaces $\delta$ by $\delta/M$ and allows any minimizer over the grid that depends only on calibration data.

\subsection{Certificates for Accepted Decisions}
\label{app:p6-selective}

Condition on the fitting split. For acceptance rule $s$ and threshold $\tau$, let $C_{i,s}$, $U_{i,s,\tau}$, and $F_{i,s}$ be the fractions of block $i$ that are accepted; accepted and close calls ($\Gamma\le\tau$, including ties) or large errors; and accepted and flipped. Empty blocks contribute zero, and $0\le F_{i,s}\le U_{i,s,\tau}\le C_{i,s}\le1$. The target is the ratio of expectations
\begin{equation}
p_s=\frac{\E F_{i,s}}{\E C_{i,s}},\qquad q_{s,\tau}=\frac{\E U_{i,s,\tau}}{\E C_{i,s}},
\end{equation}
which is the failure rate among accepted decisions when a block and then a decision within it are drawn uniformly; it is not the average of within-block ratios, and it requires $\E C_{i,s}>0$.

The proof of Theorem~\ref{thm:heldout} fixes $q=q_{s,\tau}$ and uses $W_i(q)=U_{i,s,\tau}-qC_{i,s}\in[-q,1-q]$ with $\E W_i(q)=0$. The interval has length one, so a one-sided Hoeffding bound and a union bound over the $G=|\mathcal G|$ grid pairs give
\begin{equation}
q_{s,\tau}\le\frac{\widehat U_{s,\tau}+\sqrt{\log(G/\delta)/(2n)}}{\widehat C_s}
\end{equation}
simultaneously whenever $\widehat C_s>0$. Replacing $U$ by $F$ and paying only for the family of acceptance rules proves Corollary~\ref{cor:direct-failure}. No independence between comparisons is assumed.

\noindent\emph{Separate-bounds variant.}\quad
A variant bounds $\E U$ from above and $\E C$ from below separately. With $\epsilon_2=\sqrt{\log(2G/\delta)/(2n)}$ it gives
\begin{equation}
q_{s,\tau}\le\frac{\widehat U_{s,\tau}+\epsilon_2}{\widehat C_s-\epsilon_2}
\end{equation}
when the denominator is positive; if only $S_0$ distinct acceptance rules occur, $2G$ may be replaced by $G+S_0$. Using $\log(G/\delta)$ for both separate events would establish only a failure probability of $2\delta$. Whenever its denominator is positive, this expression is a looser consequence of the event behind~\eqref{eq:certificate}.

\noindent\emph{Variance-adaptive variant.}\quad
Let $n\ge2$, let $\widehat V_W(t)$ be the unbiased sample variance of $W_i(t)=U_i-tC_i$, and define
\begin{align}
L&=\log(2G/\delta),\qquad a=\sqrt{2L/n},\qquad b=\frac{7L}{3(n-1)},\\
H(t)&=\widehat U-t\widehat C+a\sqrt{\widehat V_W(t)}+b,\qquad t\in[0,1].
\end{align}
Set
\begin{equation}
B_{\mathrm{EB}}=
\begin{cases}
\inf\{t\in[0,1]:H(t)\le0\},&\text{if the set is nonempty},\\
1,&\text{otherwise}.
\end{cases}
\label{eq:app-eb-ratio}
\end{equation}
Then $q_{s,\tau}\le B_{\mathrm{EB},s,\tau}$ simultaneously over the grid with probability at least $1-\delta$.

At the population value $q$, the empirical Bernstein inequality for variables of range one gives $\Pp\{H(q)<0\}\le\delta/G$. The function $H$ need not be monotone, but it crosses zero only once. Write $s_U=\sqrt{\widehat V_U}$ and $\eta=b+\widehat U-as_U$. Since $0\le U_i\le1$,
\begin{equation}
s_U^2\le\frac{n}{n-1}\widehat U,\qquad as_U-\widehat U\le\frac{L}{2(n-1)},
\end{equation}
so $\eta\ge11L/[6(n-1)]>0$. For $0<r<t\le1$ put $\lambda=t/r$. From $W(t)=\lambda W(r)-(\lambda-1)U$ and the triangle inequality for sample standard deviations,
$H(t)\le\lambda H(r)-(\lambda-1)\eta$,
so $H(r)\le0$ implies $H(t)<0$ for every $t>r$. Since $H(0)>0$ and $H$ is continuous, the first crossing in~\eqref{eq:app-eb-ratio} is a valid upper confidence endpoint, and no grid over $t$ is needed. Numerically, if $H(1)\ge0$ the implementation returns one; otherwise bisection returns the upper end of the bracket.

\subsection{Choosing Among Certificates}

If the block budget is split before observation into selection blocks and certification blocks, a grid pair and a certificate may be chosen on the first part and, given that choice, certified on the second with the single-rule width $\sqrt{\log(1/\delta)/(2n_{\mathrm{cert}})}$ and no grid factor. Choosing again after viewing several certification results requires a simultaneous guarantee. Likewise, each prespecified method and sample size carries its own guarantee, and selecting the smallest certificate along a curve requires allocating the confidence across the points eligible for selection.

% ============================================================
\section{Other Quantizers}
\label{app:p7}

Different quantizer families induce different comparison errors under the common ranking and pruning interface.

\subsection{RaBitQ}

For a fixed unit data direction $o$ and query direction $q$, RaBitQ uses
\begin{equation}
\widehat s_R(o,q)=\frac{\langle\bar o(P),q\rangle}{\langle\bar o(P),o\rangle},
\end{equation}
whose denominator is positive for every orthogonal $P$, because $\langle\bar o(P),o\rangle=\|P^\top o\|_1/\sqrt D\ge1/\sqrt D$. With a fixed tie rule, the event that the error exceeds the random radius is measurable. A shared rotation couples the distances of a comparison, so ranking and pruning statements go through the joint event, the triangle inequality, and a union bound rather than independence, and conditioning on the realized radius does not create conditional coverage. For raw squared distance,
$\xi_R^{d^2}(u_r,v_r)=-2r_ur_v\,\varepsilon_R(u,v)$,
and a pruning rule stated in Euclidean distance becomes an $\alpha^2$ rule after squaring, with a different radial factor for each of the two distances. A rule stated directly in squared distance or cosine dissimilarity keeps its own parameter.

\subsection{BBQ and Block Estimators}

Lucene BBQ scores in two directions that both shape the graph: stored vectors are one-bit, and temporary int4 query-side vectors support construction. Candidate collection, diversity pruning, and reverse-link scoring each need their own calibration of the error. Our reference scorer implements these scoring directions. Product quantizers and related block estimators enter through sums of block errors, whose tails come from bounded blocks, block covariance, or held-out calibration. The formulation assumes a floating-point reference metric; embeddings that are binary by construction have no such reference.

\section{Datasets, Protocols, and Experiment Inventory}
\label{app:experiments}

% TODO：上传实验代码后核验预留仓库链接与复现附件清单。
\noindent\url{https://github.com/Astrmira/low-bit-decision-stability}

Decision experiments use bounded memory-mapped samples, and E8 builds complete 200{,}000-vector graphs. Fixed seeds, disjoint pools, and independent block streams retain the stated fitting, certification, and validation roles. The repository pairs programs with frozen results and hashes. Scripts \texttt{calibration\_cross.py}, \texttt{matched\_fidelity.py}, \texttt{diagnostic\_reanalysis.py}, and \texttt{observable\_certificates.py} reproduce the controls. E8 retains graphs, labels, traces, clocks, fitted constants, and source/input identities. Wolt evidence includes three instrumented rebuilds, 12 entry-prune replays, 72 fixed-pool variants, and seven entry or adjacency interventions with byte-level patch manifests. Independent checks verify scores, selection outputs, event counts, and cluster bootstraps.

\begin{table}[!htbp]
\centering
\small
\setlength{\tabcolsep}{3pt}
\begin{tabular}{lllrrrrrrr}
\toprule
& & & & \multicolumn{4}{c}{Recall@10 (\%)} & \multicolumn{2}{c}{Native search} \\
Data & Order & Axes & Width & Native & SM & FP32 & Cross & L0 scores & QPS \\
\midrule
Cohere & A+B & Orig. & 64 & 93.29 & 94.86 & 96.95 & 90.17 & 2,386 & 2,726 \\
Cohere & A+B & Orig. & 128 & 95.69 & 95.34 & 98.69 & 91.82 & 3,920 & 1,715 \\
Cohere & A+B & Orig. & 256 & 96.10 & 95.50 & 99.50 & 92.19 & 6,602 & 984 \\
Cohere & A+B & Rot. & 64 & 86.23 & 84.95 & 90.24 & 89.81 & 1,991 & 3,172 \\
Cohere & A+B & Rot. & 128 & 88.70 & 87.62 & 93.71 & 92.24 & 3,189 & 2,062 \\
Cohere & A+B & Rot. & 256 & 89.96 & 89.05 & 95.47 & 93.98 & 5,194 & 1,250 \\
\midrule
Wolt & B & Orig. & 64 & 88.23 & 87.45 & 91.75 & 82.39 & 1,727 & 2,588 \\
Wolt & B & Orig. & 128 & 89.94 & 89.37 & 93.54 & 84.76 & 2,539 & 1,902 \\
Wolt & B & Orig. & 256 & 90.70 & 90.62 & 95.08 & 86.54 & 3,867 & 1,271 \\
Wolt & B & Rot. & 64 & 86.44 & 82.63 & 90.55 & 84.87 & 1,667 & 2,596 \\
Wolt & B & Rot. & 128 & 87.47 & 86.14 & 92.45 & 88.05 & 2,490 & 2,013 \\
Wolt & B & Rot. & 256 & 87.31 & 87.08 & 93.73 & 89.95 & 3,851 & 1,212 \\
\midrule
Wolt & A & Orig. & 64 & 0.00 & 0.00 & 0.00 & 0.00 & 68 & 24,649 \\
Wolt & A & Orig. & 128 & 0.00 & 0.00 & 0.00 & 0.00 & 68 & 22,937 \\
Wolt & A & Orig. & 256 & 0.00 & 0.00 & 0.00 & 0.00 & 68 & 23,068 \\
Wolt & A & Rot. & 64 & 0.00 & 0.00 & 0.00 & 0.00 & 68 & 26,773 \\
Wolt & A & Rot. & 128 & 0.00 & 0.00 & 0.00 & 0.00 & 68 & 25,182 \\
Wolt & A & Rot. & 256 & 0.00 & 0.00 & 0.00 & 0.00 & 68 & 26,273 \\
\bottomrule
\end{tabular}
\caption{Complete E8 width sweep on the eight 200{,}000-vector graphs. Scoring controls use the same graph, entry, and FP32 budget 64. Work and five-round single-CPU QPS refer to native search. The last block has 68 reachable nodes.}
\label{tab:actual-graph}
\end{table}

% ============================================================
\section{Representation Diagnostics}
\label{app:foundational}

The diagnostics in this appendix motivated the Gaussian model of Section~\ref{sec:bridge} and the applicability check of Section~\ref{sec:applicability}. Table~\ref{tab:app-coordinate-qq} reports directly recomputed coordinate fits.

\begin{table}[!htbp]
\centering
\small
\begin{tabular}{lrrr}
\toprule
Dataset & $D$ & Mean coordinate QQ $R^2$ & Norm CV \\
\midrule
Cohere & 768 & .9984 & .0439 \\
MiniLM & 384 & .9993 & $3.31\times10^{-8}$ \\
Wolt-CLIP & 512 & .9969 & $3.29\times10^{-8}$ \\
GIST & 960 & .9181 & $3.31\times10^{-8}$ \\
\bottomrule
\end{tabular}
\caption{Direct representation diagnostics on 4{,}000 archived vectors per dataset. Each coordinate is sorted and regressed, with an intercept, on normal Blom quantiles $(i-0.375)/(n+0.25)$. The table averages coordinate-wise $R^2$. Cohere retains its archived radii. The other three inputs are unit-normalized, so their norm CV describes that input convention. These are coordinate QQ fits, distinct from the sign-entropy fits below.}
\label{tab:app-coordinate-qq}
\end{table}

Coarse geometry summaries vary across representations. The fraction of principal components explaining 90\% of variance is 0.309 for Cohere, 0.497 for MiniLM, and 0.181 for GIST. Their representative nearest-neighbour angular gaps are $0.125^\circ$, $0.277^\circ$, and $0.091^\circ$. The analysis uses decision margins and the errors relative to them.

\begin{table}[tbp]
\centering

\small
\begin{tabular}{lrrrrr}
\toprule
Dataset & Coord.\ sign & Random hyperplane & GW & KL & $H_{\mathrm{sign}}$ \\
\midrule
Cohere & .650 & .744 & .746 & .0229 & .747 \\
BGE-M3 & .677 & .701 & .692 & .0022 & .700 \\
MiniLM & .508 & .506 & .506 & .0021 & .987 \\
GIST & .9999 & .768 & .774 & .2648 & .0004 \\
Random & .513 & .513 & .513 & .0010 & .981 \\
\bottomrule
\end{tabular}

\small
\caption{Coordinate signs versus random hyperplanes: collision rate of coordinate signs, collision rate of random-hyperplane signs, the angular prediction $1-\theta/\pi$ (GW), the KL divergence between the two collision laws, and the sign entropy $H_{\mathrm{sign}}$.}
\label{tab:app-foundational-sign}
\end{table}

GIST is the key negative control: almost all coordinate signs agree whatever the angle between two vectors, whereas random hyperplanes keep the angular collision law. This is the origin of both the effect of rotation and the applicability check.

\begin{table}[tbp]
\centering

\small
\begin{tabular}{lrrrrr}
\toprule
Dataset & $D$ & $R^2$ & MAE & Mean $|\mathrm{SNR}|$ & Measured / predicted entropy \\
\midrule
Cohere & 768 & .9996 & .004 & .884 & .747 / .748 \\
BGE-M3 & 1024 & .9989 & .007 & .908 & .700 / .699 \\
MiniLM & 384 & .9222 & .002 & .127 & .987 / .988 \\
Random & 768 & .9967 & .001 & .164 & .981 / .982 \\
GIST & 960 & not fitted & .240 & 1.763 & .000 / .240 \\
MSMARCO-Cohere & 1024 & .9995 & .0019 & .370 & .910 / .910 \\
\bottomrule
\end{tabular}

\small
\caption{Gaussian model of sign entropy: the entropy of each coordinate predicted from its signal-to-noise ratio $\mu_i/\sigma_i$. $R^2$ is unstable when the entropy barely varies across coordinates, so the mean absolute error (MAE) is reported alongside.}
\label{tab:app-foundational-gaussian}
\end{table}

For the sign-entropy prediction fits in nine 768-dimensional regimes, eight had $R^2\ge0.9965$; the near-isotropic RoBERTa regime had a low $R^2$ because its entropy barely varies across coordinates, with a mean absolute error of 0.00166. These checks are coordinate-wise; the joint test of Section~\ref{sec:exp-eligibility} is stronger.

% ============================================================
\section{Representation Mechanisms}
\label{app:mechanisms}

Tables~\ref{tab:app-magnitude} to~\ref{tab:app-rotation-drivers} collect the measurements behind Section~\ref{sec:bridge}: the gain from the magnitude bit, a controlled change of coordinate heterogeneity, the response of the two-bit code to rotation, and the predictors of that response.

\begin{table}[tbp]
\centering

\small
\begin{tabular}{lrrr}
\toprule
Dataset & $\mathrm{CV}(\sigma)$ & $\Delta F$ & $\Delta$ Recall \\
\midrule
MiniLM & .118 & +.132 & +.174 \\
Cohere & .182 & +.091 & +.144 \\
CodeSearch & .098 & +.088 & +.143 \\
Landmark & .110 & +.088 & +.110 \\
Arxiv & .108 & +.071 & +.159 \\
Random & .070 & +.049 & +.210 \\
\bottomrule
\end{tabular}

\small
\caption{Gain from the magnitude bit in distance rank correlation ($\Delta F$) and in simulated recall. $\mathrm{CV}(\sigma)$ is the coefficient of variation of the coordinate standard deviations.}
\label{tab:app-magnitude}
\end{table}

\begin{table}[tbp]
\centering

\small
\begin{tabular}{lrrrr}
\toprule
Intervention & $\mathrm{CV}(\sigma)$ & $F_{1\mathrm{bit}}$ & $F_{2\mathrm{bit}}$ & $\Delta F_{\mathrm{mag}}$ \\
\midrule
Amplify, $\gamma=2$ & .316 & .682 & .719 & .036 \\
Amplify, $\gamma=1.5$ & .252 & .694 & .728 & .034 \\
Original & .182 & .706 & .738 & .032 \\
Whiten, $.50$ & .091 & .720 & .746 & .026 \\
Flatten & .002 & .736 & .757 & .022 \\
\bottomrule
\end{tabular}

\small
\caption{Controlled change of coordinate heterogeneity on Cohere. The rank correlations of the codes and the gain from the magnitude bit move in opposite directions.}
\label{tab:app-intervention}
\end{table}

\begin{table}[tbp]
\centering

\small
\begin{tabular}{lrrl}
\toprule
Dataset & Sign entropy before $\to$ after & Recall response & Regime \\
\midrule
GIST & .000 $\to$ .511 & +307\% & signs nearly fixed \\
Wolt-CLIP & .836 $\to$ .616 & +3.2 pp & over-spread \\
Cohere & .747 $\to$ .563 & $-0.5$ pp & coordinate signal \\
MiniLM & approximately unchanged & approximately zero & near-isotropic \\
\bottomrule
\end{tabular}

\small
\caption{Response of the two-bit code to rotation. Recall changes are taken from the original experiment reports.}
\label{tab:app-rotation}
\end{table}

\begin{table}[tbp]
\centering

\small
\begin{tabular}{lrr}
\toprule
Predictor & Spearman $\rho$ & $p$ \\
\midrule
$\mathrm{CV}(\sigma)$ & $+.657$ & $.020$ \\
$H_{\mathrm{sign}}$ & $-.909$ & $<.001$ \\
$1-H_{\mathrm{sign}}$ & $+.909$ & $<.001$ \\
$(1-H_{\mathrm{sign}})\times\mathrm{CV}(\sigma)$ & $\mathbf{+.930}$ & $<.001$ \\
$\log D$ & $-.080$ & $.805$ \\
\bottomrule
\end{tabular}

\small
\caption{Predictors of the change $\Delta F_{\mathrm{rot}}$ in the two-bit code's distance rank correlation under rotation, over 12 datasets.}
\label{tab:app-rotation-drivers}
\end{table}

Excluding the one non-Gaussian outlier, Landmark-Nomic, raises the correlation of the entropy gap to 0.964 and that of heterogeneity to 0.746; the response does not depend on dimension. In the heterogeneity intervention, the correlation between heterogeneity and the gain from the magnitude bit was $+1.0$ in four of five datasets.

% ============================================================
\section{Details for E1: Margins Predict Flips}
\label{app:ranking}

\begin{table}[tbp]
\centering

\small
\begin{tabular}{llrrrr}
\toprule
Dataset & Code & Passes check & Flip rate & Query failure & Near-tie exponent \\
\midrule
Cohere & 1-bit & yes & 8.89\% & 63.12\% & 1.526 \\
Cohere & 2-bit & yes & 5.83\% & 47.50\% & 1.526 \\
MiniLM & 1-bit & yes & 11.21\% & 59.38\% & 1.500 \\
MiniLM & 2-bit & yes & 4.19\% & 43.75\% & 1.500 \\
GIST & 1-bit & no & 100.00\% & 100.00\% & 1.595 \\
GIST & 2-bit & no & 35.83\% & 91.88\% & 1.595 \\
\bottomrule
\end{tabular}

\small
\caption{Nearest-vs-rest decisions on 32-neighbour sets. ``Passes check'' uses the first three criteria of Section~\ref{sec:applicability}. Query failure means at least one comparison flips.}
\label{tab:app-probe24}
\end{table}

\begin{table}[tbp]
\centering

\small
\setlength{\tabcolsep}{4pt}
\begin{tabular}{lrrrrr}
\toprule
MNR & Cohere 1-bit & Cohere 2-bit & MiniLM 1-bit & MiniLM 2-bit & GIST 2-bit \\
\midrule
$[0,.25)$    & 33.3\% (93)  & 48.6\% (70)   & 48.1\% (104) & 50.0\% (72)   & 41.1\% (158) \\
$[.25,.5)$   & 31.3\% (150) & 35.3\% (102)  & 36.0\% (125) & 39.6\% (53)   & 37.5\% (301) \\
$[.5,.75)$   & 30.2\% (182) & 24.8\% (121)  & 37.4\% (227) & 29.5\% (78)   & 39.2\% (551) \\
$[.75,1)$    & 21.3\% (286) & 20.3\% (158)  & 24.7\% (312) & 23.9\% (138)  & 35.7\% (762) \\
$[1,1.5)$    & 13.2\% (832) & 11.9\% (520)  & 15.8\% (991) & 11.5\% (355)  & 28.2\% (1389) \\
$[1.5,2)$    & 8.49\% (1083)& 7.68\% (690)  & 10.2\% (1064)& 5.95\% (571)  & 34.2\% (803) \\
$[2,3)$      & 3.02\% (1458)& 2.54\% (1577) & 2.53\% (1227)& 1.29\% (1393) & 40.8\% (622) \\
$[3,5)$      & 0.12\% (804) & 0.14\% (1413) & 0.28\% (718) & 0.13\% (1531) & 42.9\% (231) \\
$[5,\infty)$ & 0.00\% (72)  & 0.00\% (309)  & 0.00\% (192) & 0.00\% (769)  & 63.6\% (143) \\
\midrule
Share with MNR $\ge2$ & 47.1\% & 66.5\% & 43.1\% & 74.5\% & 20.1\% \\
Share with MNR $<1$   & 14.3\% & 9.1\% & 15.5\% & 6.9\% & 35.7\% \\
\bottomrule
\end{tabular}

\small
\caption{Flip rates by MNR for $4{,}960$ comparisons per configuration, with counts in parentheses. Figure~\ref{fig:margin-bins} pools the three two-bit columns into six bins. GIST fails the applicability check and shows no decreasing trend.}
\label{tab:app-margin-bins}
\end{table}

\subsection{Predictor Robustness and Components}
\label{app:predictor}

All statistics in Table~\ref{tab:predictor-controls} use saved calibration summaries. The margin quantile is the 20th percentile, recorded as the threshold for 80\% retention, and $v$ is the calibrated comparison-error standard deviation. For full score minus margin quantile, paired dataset-cluster 95\% intervals are $[.041,.730]$ for ranking and $[.014,.263]$ for pruning. The full-score-minus-MNR intervals appear in Section~\ref{sec:exp-margins}. Each bootstrap uses 2{,}000 replicates with seed 20261109 and resamples the twelve dataset clusters, keeping both codes together. Leave-one-dataset-out analyses likewise remove both codes.

Table~\ref{tab:predictor-controls} compares the direction of the original-to-rotated change within each dataset. The denominator excludes exact ties in held-out flip rates, occurring only for sphere pruning. All three statistics are oriented toward increasing risk. Across configurations, the full score exceeds observed flips by 25.49 percentage points on average for ranking (range 7.96--61.33) and 23.60 for pruning (2.04--70.54).

The second control separates the two terms of the bound with the original three disjoint pools. The first pool of 2{,}000 vectors fixes the distance-fitted calibration on 48 target blocks and fixes $\tau$ on the original nine-quantile grid. The second pool of 2{,}000 vectors and 96 target blocks estimates
\[
\widehat B(\tau)=\widehat{\Pr}\{g\le\tau\},\qquad
\widehat T(\tau)=\widehat{\Pr}\{|Z|\ge\tau\},\qquad
\widehat R(\tau)=\min\{1,\widehat B(\tau)+\widehat T(\tau)\},
\]
where $g$ is the non-negative exact margin and $Z$ its comparison error. The prediction target is the flip rate on the prespecified fifth of decisions with the smallest margins in the third pool, which also contains 2{,}000 vectors and 96 target blocks. The second pool serves only to estimate the terms and the third only to measure the outcome; scoring, thresholds, and the hardest fraction stay fixed.

\begin{table}[tbp]
\centering

\small
\begin{tabular}{lrrrr}
\toprule
Predictor & Ranking $\rho$ & 95\% interval & Pruning $\rho$ & 95\% interval \\
\midrule
Close-call probability $\widehat B$ & 0.627 & [0.131, 0.932] & 0.969 & [0.837, 0.991] \\
Large-error probability $\widehat T$ & 0.953 & [0.788, 0.991] & 0.890 & [0.632, 0.964] \\
Combined $\widehat R$ & 0.955 & [0.749, 0.989] & 0.983 & [0.908, 0.992] \\
\bottomrule
\end{tabular}

\small
\caption{Second-pool components predicting third-pool hardest-fifth flip rates over 24 configurations. Intervals are 95\% dataset-cluster bootstrap intervals retaining both rotations.}
\label{tab:app-direct-components}
\end{table}

For ranking, the paired intervals for the combined estimate minus the close-call term and minus the large-error term are $[.013,.812]$ and $[-.092,.097]$; for pruning they are $[-.047,.111]$ and $[.008,.297]$. The strongest single term therefore differs between decision types. Excluding the configuration with a nonpositive slope gives combined correlations of 0.949 and 0.981, and leaving out one dataset cluster gives 0.942 to 0.979 for ranking and 0.979 to 0.986 for pruning.

For Cohere, MiniLM, and GIST we also reconstruct the first-pool decisions from the original sampled vector indices. The selected thresholds, combined estimates, error scales, and decision counts reproduce the saved results exactly for both codes and both decision types. Table~\ref{tab:app-calibration-components} reports the frequencies of the two terms.

\begin{table}[tbp]
\centering

\small
\begin{tabular}{llrrrrrr}
\toprule
& & \multicolumn{3}{c}{Ranking} & \multicolumn{3}{c}{Pruning} \\
Dataset & Code & $\widehat B$ & $\widehat T$ & $\widehat R$ & $\widehat B$ & $\widehat T$ & $\widehat R$ \\
\midrule
Cohere & 2-bit & 18.00 & 13.92 & 31.91 & 35.00 & 20.35 & 55.35 \\
 & Rotated & 18.00 & 12.63 & 30.63 & 35.00 & 5.72 & 40.72 \\
MiniLM & 2-bit & 12.01 & 10.24 & 22.25 & 12.01 & 1.37 & 13.38 \\
 & Rotated & 12.01 & 9.13 & 21.14 & 12.01 & 1.34 & 13.34 \\
GIST & 2-bit & 18.00 & 58.02 & 76.02 & 35.00 & 29.39 & 64.39 \\
 & Rotated & 18.00 & 23.40 & 41.40 & 35.00 & 5.97 & 40.97 \\
\bottomrule
\end{tabular}

\small
\caption{Replay of the first-pool calibration at the original fixed thresholds, in percent. Each row reproduces the original combined estimate for both decision types.}
\label{tab:app-calibration-components}
\end{table}

\subsection{Calibration Direction and Same-Dataset Code Choice}
\label{app:calibration-direction}

The E1 grid is the $.02,.04,.08,.12,.18,.25,.35,.50,.70$ quantiles of positive fitting margins. Ranking orients the exact winner against each alternative. Pruning uses signed $\Psi=d(t,c)-\alpha d(s,c)$, with nonnegative values pruned. The score $\widehat r_2$ uses $|R|$, while $\widehat r_{\mathrm{opp}}$ uses $-\sign(\Psi)R$. Exact ties enter the close-call mass. Both use empirical tails. Appendix~\ref{app:covariance} evaluates the one-sided Gaussian plug-in.

Distance fitting estimates $b,a$ from individual distance pairs. In E1 these are target--candidate pairs from the fitting pools. E2's distance control uses independent random pairs and its comparison fit uses the original 80 blocks. E8's original diagnostic fits comparison differences through the origin on 128 queries. Its distance control fits individual distances in the same FP32-visited pools. For the E1 and six-code comparison controls, set $\widehat\Psi_C=g\widehat\Psi_D$, with
\begin{equation}
g=\frac{\sum_{\mathrm{fit}}\widehat\Psi_D\Psi}{\sum_{\mathrm{fit}}\widehat\Psi_D^2}.
\end{equation}
A positive $g$ preserves all comparison decisions. The residual $\widehat\Psi_C-\Psi$ and its empirical ordering score change with the fitted map. Nonpositive gains receive score one. Every map and threshold is fixed on fitting blocks, and all four fitting/event variants are reported.

On the 160-query Cohere set, distance-fitted $\widehat r_2$ falls from .659 to .378 under rotation while flips rise from 5.83\% to 8.81\%. The opposing score rises from .146 to .291. Comparison fitting gives .309 to .372 for $\widehat r_2$ and .265 to .325 for the opposing score. In the E1 48-candidate ranking pools, distance-fitted scores change from .319 to .306 (two-sided) and .227 to .301 (opposing), with flips 4.43\% to 4.48\%. For E1 pruning, the two-sided score falls from .553 to .407 as flips fall from 10.25\% to 6.93\%, whereas the distance-fitted opposing score changes from .358 to .367. The release records all 48 calibration positions in \texttt{calibration-comparison.csv}.

The six-code experiment uses the four data sources and code implementations of Section~\ref{sec:instantiations} with shared exact-neighbour pools. A fixed permutation of 4{,}000 rows assigns 1{,}200 to code fitting, 800 to calibration, and 2{,}000 to evaluation. All codes share 24{,}000 calibration distance pairs, 96 calibration targets, 160 evaluation targets, exact 48-neighbour pools, $\alpha=1.2$, and degree 24. The code-fitting rows determine centroids, scalar ranges, and PQ codebooks. Calibration fixes the maps and thresholds. For constant SIFT/GIST sign scores, the distance estimate equals the calibration mean. Ranking uses lexicographic ties and pruning uses the nonnegative action convention. Five hundred shared evaluation-target bootstraps keep fitting statistics fixed. Local controls use unique directed endpoints of these fitting decisions with frozen maps. Endpoint-frequency-weighted controls are retained in \texttt{local-distance-controls.csv}.

\begin{table}[tbp]
\centering
\small
\setlength{\tabcolsep}{3pt}
\begin{tabular}{llrrrrrrr}
\toprule
Data & Role & $D_2$ & $D_o$ & $C_2$ & $C_o$ & MNR & MSE G/L & Fidelity G/L \\
\midrule
Cohere & Rank & 0.600 & 0.943 & 1.000 & 1.000 & 0.314 & 0.714/0.886 & 0.714/0.143 \\
Cohere & Prune & 0.771 & 1.000 & 0.943 & 0.943 & 0.943 & 0.714/0.714 & 0.714/0.371 \\
MiniLM & Rank & 0.943 & 0.943 & 0.943 & 0.943 & 0.943 & 0.943/0.943 & 0.943/1.000 \\
MiniLM & Prune & 0.943 & 0.943 & 0.943 & 0.943 & 0.943 & 0.943/0.943 & 0.943/0.943 \\
SIFT & Rank & 1.000 & 0.943 & 1.000 & 1.000 & 0.657 & 1.000/0.943 & 1.000/0.943 \\
SIFT & Prune & 0.829 & 0.371 & 0.829 & 0.943 & 0.943 & 0.657/0.657 & 0.657/0.829 \\
GIST & Rank & 1.000 & 0.943 & 1.000 & 1.000 & 0.829 & 0.943/0.943 & 0.943/0.943 \\
GIST & Prune & 0.829 & 0.943 & 0.829 & 0.943 & 0.829 & 0.600/0.600 & 0.600/0.600 \\
\bottomrule
\end{tabular}
\caption{Within-dataset Spearman correlations across six codes. $D/C$ fit distances/comparisons, and $2/o$ use two-sided/opposing errors. MNR is $-q_{.2}/v$. MSE and $1-\rho_{\mathrm{distance}}$ give global/local endpoint controls (G/L). All metrics use calibration data. Query-cluster intervals and all 15 code-pair comparisons are retained.}
\label{tab:app-six-code-choice}
\end{table}

% ============================================================
\section{Details for E2: Correlated Errors}
\label{app:covariance}

\noindent\emph{Close-call plug-in.}\quad
For Table~\ref{tab:covariance}, the positive comparison fit uses 80 query blocks (2{,}480 comparisons) and fixes $a$, $\widehat\mu_R$, and $\widehat\nu^2$. It selects $\widehat\tau$ by minimizing the empirical close-call mass plus the one-sided tail over fitting-margin quantiles $.02,.04,.08,.12,.18,.25,.35,.50$. The 160 held-out blocks (4{,}960 positive-margin comparisons) then give
\begin{equation}
\widehat B=\min\!\left\{1,\frac1N\sum_{i=1}^N\mathbf1\{\Gamma_i\le\widehat\tau\}
+\exp\!\left[-\frac{(\widehat\tau+\widehat\mu_R)_+^2}{2\widehat\nu^2}\right]\right\}.
\label{eq:app-covariance-plugin}
\end{equation}
Cohere two-bit codes have $a=2.71821\times10^{-4}$, $c_Q=1$ in metric units, $\widehat\mu_R=-0.00354069$, $\widehat\nu^2=0.000178638$, and $\widehat\tau=0.02777196$. The two terms are $0.360887+0.193318=0.554205$. Replacing the variance by the sum of marginal variances, $0.00132533$, and keeping the threshold fixed gives one after clipping. These are empirical plug-ins of the close-call and one-sided-tail formulas. Block certificates use the separate procedure of Section~\ref{sec:heldout}.

\noindent\emph{Selection regimes.}\quad
The experiment of Table~\ref{tab:selection-covariance} keeps one distance-fitted calibration fixed and varies how comparisons are generated, using 6{,}000 sampled vectors, 60{,}000 independent fitting pairs, 160 query blocks, and 128 candidate pairs per query. The independent control samples all three endpoints independently with replacement, and a second control rejects repeated endpoints; the two differ by at most 0.003. In every regime, the within/between identity~\eqref{eq:app-anova-covariance} reproduces the pooled covariance and the variance of the difference to floating-point precision. Query-block bootstrap intervals accompany the grouped regimes, together with the 5th, 50th, and 95th percentiles of the within-query covariance, since a small weighted average need not describe every query. For nearest-vs-rest comparisons the nearest neighbour's error is constant within a query, so the within-query correlation is undefined.

\noindent\emph{Calibration and sampling.}\quad
The original covariance experiment samples 4{,}000 vectors, fits on 80 query blocks, and measures flips on 160 further blocks. The selection experiment uses a separate 6{,}000-vector population and 160 query blocks. Table~\ref{tab:app-crossed-calibration} evaluates both fixed calibrations on both axes and all three populations, with coordinate-wise FP64 accumulation. The rotated distance fit uses 60{,}000 independent random pairs in the same rotated coordinates. An intercept cancels in the ranking residual, while the fitted slope changes its covariance and bias.

\begin{table}[tbp]
\centering
\small
\setlength{\tabcolsep}{3.4pt}
\begin{tabular}{lllrrrrrrr}
\toprule
Population & Axis & Fit & $10^4a$ & $\rho$ & $I$ & $10^4C$ & $10^4V_\Delta$ & $10^2\mu_R$ & $10^2v/c_Q$ \\
\midrule
4K, fit 80 & Orig. & Comp. & 2.718 & 0.870 & 7.42 & 5.733 & 1.786 & -0.354 & 1.337 \\
4K, fit 80 & Orig. & Dist. & 4.460 & 0.584 & 2.36 & 3.511 & 5.162 & +1.795 & 2.272 \\
4K, fit 80 & Rot. & Comp. & 4.769 & 0.435 & 1.77 & 0.845 & 2.195 & -0.531 & 1.481 \\
4K, fit 80 & Rot. & Dist. & 5.180 & 0.405 & 1.68 & 0.858 & 2.524 & -0.257 & 1.589 \\
\midrule
4K, held-out 160 & Orig. & Comp. & 2.718 & 0.824 & 5.31 & 4.361 & 2.023 & -0.414 & 1.422 \\
4K, held-out 160 & Orig. & Dist. & 4.460 & 0.524 & 2.08 & 3.176 & 5.875 & +1.607 & 2.424 \\
4K, held-out 160 & Rot. & Comp. & 4.769 & 0.459 & 1.85 & 0.875 & 2.066 & -0.743 & 1.437 \\
4K, held-out 160 & Rot. & Dist. & 5.180 & 0.427 & 1.74 & 0.881 & 2.368 & -0.499 & 1.539 \\
\midrule
6K, 160 & Orig. & Comp. & 2.718 & 0.843 & 6.30 & 4.505 & 1.701 & -0.590 & 1.304 \\
6K, 160 & Orig. & Dist. & 4.460 & 0.613 & 2.59 & 3.960 & 4.995 & +1.268 & 2.235 \\
6K, 160 & Rot. & Comp. & 4.769 & 0.463 & 1.86 & 0.879 & 2.041 & -0.584 & 1.429 \\
6K, 160 & Rot. & Dist. & 5.180 & 0.436 & 1.77 & 0.901 & 2.336 & -0.333 & 1.528 \\
\bottomrule
\end{tabular}
\caption{The full calibration cross for Cohere two-bit codes, using the same signed-DCT rotation in all three populations. $a$ maps native scores to exact-distance units, in which $c_Q=1$. In native units $c_Q=1/a$, $\mu_{R,\mathrm{raw}}=\mu_R/a$, and $v_{\rm raw}/c_Q=v_{\rm metric}$. $I$ is the cancellation factor, $C$ the covariance, and $V_\Delta$ the comparison-error variance. Comp.\ fits comparisons through the origin. Dist.\ fits 60{,}000 independent random pairs. All fitted constants stay frozen across populations. Standardized bias is $\mu_R/v$ in these metric units, computed from the bias and effective-noise columns.}
\label{tab:app-crossed-calibration}
\end{table}

For the comparison-fitted calibration on the 80-query sample, the within-query variance components of the two distance errors are $0$ and $0.967\times10^{-4}$, and the between-query components are $5.909\times10^{-4}$ and $6.377\times10^{-4}$. The between-query covariance is $5.733\times10^{-4}$, and the variance of the difference splits into $0.967\times10^{-4}$ within queries and $0.820\times10^{-4}$ between queries. The fixed nearest neighbour makes its error constant within a query, which explains the zero within-query covariance. An intercept alone leaves the covariance unchanged, so the effect of the calibration acts through the fitted slope.

On the 160 held-out queries, the unrotated cancellation factor is 5.31 under comparison fitting and 2.08 under distance fitting. Rotation lowers these to 1.85 and 1.74. Under distance fitting, the residual standard deviation falls from 0.02424 to 0.01539, but its bias shifts from $+0.01607$ to $-0.00499$. The exact signed-margin distribution is unchanged, and the flip rate rises from 5.83\% to 8.81\% under either positive calibration. Bias, scale, and the joint distribution of error and margin enter the event together.

% ============================================================
\section{Details for E3: First Divergence}
\label{app:trace}

\begin{table}[tbp]
\centering

\small
\setlength{\tabcolsep}{5.5pt}
\begin{tabular}{llrrrr}
\toprule
Dataset & Code & Pruning flip rate & Outputs differ & Jaccard, fixed & Jaccard, re-sorted \\
\midrule
Cohere & 2-bit & 13.06\% & 100.0\% & .322 & .395 \\
Cohere & rotated 2-bit & 13.07\% & 100.0\% & .322 & .289 \\
MiniLM & 2-bit & 1.92\% & 100.0\% & .654 & .529 \\
MiniLM & rotated 2-bit & 1.92\% & 100.0\% & .656 & .529 \\
GIST & 2-bit & 7.91\% & 99.4\% & .458 & .260 \\
GIST & rotated 2-bit & 7.90\% & 99.4\% & .458 & .359 \\
\bottomrule
\end{tabular}

\small
\caption{RobustPrune without backfill on exact candidate pools ($\alpha=1.2$, $\Delta=32$, 64 candidates, 160 targets) with scale-only calibration. ``Outputs differ'' counts calls with at least one action mismatch under the fixed exact order; the re-sorted replay also includes the change of order.}
\label{tab:app-probe27}
\end{table}

With the scale-only calibration, every observed pruning flip is a missed prune; the matched calibration controls below separate the role of the intercept. In this experiment, rotation leaves the fixed-order pruning statistics nearly unchanged, while the re-sorted replay also reflects the change in candidate order.

\begin{table}[tbp]
\centering

\small
\setlength{\tabcolsep}{3pt}
\begin{tabular}{lrrrrrrr}
\toprule
Data & Edges & Actions & Loss & All-pairs & Blocker & $\ge1$ & $\rho$ \\
\midrule
Cohere & 3{,}054 & 24.7 & 34.3\% & 92.96 & 18.03 & 89.0\% & .81 \\
MiniLM & 5{,}118 & 21.0 & 21.5\% & 41.91 & 12.71 & 89.4\% & .69 \\
GIST & 3{,}701 & 21.7 & 24.7\% & 131.75 & 18.23 & 91.6\% & .73 \\
\bottomrule
\end{tabular}

\small
\caption{Two-bit RobustPrune edge scores. Actions is the mean prefix-action count and Loss is fixed-order edge loss. All-pairs and Blocker are mean sums from the coarse double sum and~\eqref{eq:action-bound}. The last two columns give blocker-sum saturation and its correlation with edge loss.}
\label{tab:app-edge-cert}
\end{table}

The edge and path tables use the blocker-aware action bounds~\eqref{eq:action-bound}. Independent replay preserves all six edge configurations' counts, failure rates, and saturation rates. Mean all-pairs sums fall by factors 3.3--7.2 under blocker aggregation, while edge saturation remains 89--92\%. All 960 candidate permutations change under quantized sorting. Original-axis inversion fractions are .291, .229, and .443 for Cohere, MiniLM, and GIST. Fixed-order and re-sorted Jaccards in Table~\ref{tab:app-probe27} measure the corresponding output changes. Full per-code edge and four-length path comparisons appear in \texttt{blocker-comparison.csv}.

\begin{table}[tbp]
\centering

\small
\begin{tabular}{lrrrrrrrr}
\toprule
& \multicolumn{2}{c}{1 hop} & \multicolumn{2}{c}{2 hops} & \multicolumn{2}{c}{3 hops} & \multicolumn{2}{c}{4 hops}\\
Dataset & lost & exceeds & lost & exceeds & lost & exceeds & lost & exceeds \\
\midrule
Cohere & 33.8 & 88.2 & 50.0 & 98.8 & 68.2 & 99.8 & 78.7 & 100.0 \\
MiniLM & 17.7 & 89.5 & 35.5 & 98.3 & 46.3 & 99.8 & 58.7 & 100.0 \\
GIST & 24.0 & 84.7 & 46.0 & 98.7 & 60.8 & 99.8 & 72.3 & 100.0 \\
\bottomrule
\end{tabular}

\small
\caption{Diversity-selected paths under unrotated two-bit codes, 600 per length. Percentages report loss of any path edge (``lost'') and blocker-aware action sums at least one (``exceeds'').}
\label{tab:app-paths}
\end{table}

Path bounds exceed one within two hops, and the observed rate of path loss grows with length as the union of its edge losses. On every path, agreement of the listed actions implied survival, as Corollary~\ref{cor:path} states.

\noindent\emph{Matched pruning calibration.}\quad
For Cohere, MiniLM, and GIST we reproduce the original 4{,}000-vector sample, 80 calibration targets, 160 evaluation targets, 64 candidates, $\alpha=1.2$, and $\Delta=32$. All fits use the same calibration pairs, consisting of target--candidate pairs and sampled candidate--candidate pairs. A positive fit through the origin, $a_0d_Q$, reproduces the original replay. An affine fit $b+ad_Q$ is compared with its slope-only counterpart $ad_Q$, which isolates the intercept at a fixed slope. Positive scaling leaves raw pruning and candidate order unchanged.

\begin{table}[tbp]
\centering

\small
\setlength{\tabcolsep}{4pt}
\begin{tabular}{lllrrrrr}
\toprule
Dataset & Code & Calibration & MP (\%) & SP (\%) & Fixed & Re-sorted & $\overline Z/s_Z$ \\
\midrule
Cohere & 2-bit & scale-only & 13.06 & 0.00 & 0.322 & 0.395 & 0.014 \\
 &  & affine, slope only & 13.06 & 0.00 & 0.322 & 0.395 & -5.467 \\
 &  & full affine & 7.85 & 0.98 & 0.497 & 0.460 & -0.428 \\
 & Rotated & scale-only & 13.07 & 0.00 & 0.322 & 0.289 & 0.100 \\
 &  & affine, slope only & 13.07 & 0.00 & 0.322 & 0.289 & -7.518 \\
 &  & full affine & 4.10 & 2.54 & 0.563 & 0.475 & 0.046 \\
\addlinespace
MiniLM & 2-bit & scale-only & 1.92 & 0.00 & 0.654 & 0.529 & 0.825 \\
 &  & affine, slope only & 1.92 & 0.00 & 0.654 & 0.529 & -3.901 \\
 &  & full affine & 0.49 & 0.63 & 0.820 & 0.590 & 0.208 \\
 & Rotated & scale-only & 1.92 & 0.00 & 0.656 & 0.529 & 0.818 \\
 &  & affine, slope only & 1.92 & 0.00 & 0.656 & 0.529 & -4.040 \\
 &  & full affine & 0.49 & 0.62 & 0.820 & 0.599 & 0.163 \\
\addlinespace
GIST & 2-bit & scale-only & 7.91 & 0.00 & 0.458 & 0.260 & 0.294 \\
 &  & affine, slope only & 7.91 & 0.00 & 0.458 & 0.260 & -3.720 \\
 &  & full affine & 6.35 & 2.34 & 0.457 & 0.239 & 0.326 \\
 & Rotated & scale-only & 7.90 & 0.00 & 0.458 & 0.359 & 0.271 \\
 &  & affine, slope only & 7.90 & 0.00 & 0.458 & 0.359 & -10.325 \\
 &  & full affine & 2.39 & 3.46 & 0.571 & 0.462 & 0.047 \\
\bottomrule
\end{tabular}

\small
\caption{Matched pruning calibration. MP and SP are missed-prune and spurious-prune rates over the exact run, with 53{,}844 comparisons on Cohere, 96{,}160 on MiniLM, and 67{,}993 on GIST per code. ``Fixed'' and ``re-sorted'' are output Jaccards under fixed and re-sorted candidate order; $\overline Z/s_Z$ is the mean comparison error over its standard deviation.}
\label{tab:app-pruning-calibration}
\end{table}

On Cohere the unrotated affine fit has $(a,b)=(3.3283\times10^{-4},-0.54483)$ and the rotated fit $(4.7478\times10^{-4},-0.65066)$. At a fixed affine slope, the intercept shifts $\widehat\Psi_\alpha$ by $(1-\alpha)b$ and leaves the standard deviation of the error unchanged. The mean error under the scale-only calibration can be small even when every flip goes in one direction, so the mean alone does not describe the distribution of actions. The full affine calibration produces both kinds of pruning error in all six configurations. It raises fixed-order overlap on Cohere, MiniLM, and rotated GIST, while unrotated GIST changes from 0.458 to 0.457. An independent implementation of integer scoring and of the selection rule reproduces all 2{,}880 fixed-order outputs and their action-mismatch events across the three calibrations.

% ============================================================
\section{Details for E4: Held-Out Certificates}
\label{app:heldout}

\begin{table}[tbp]
\centering

\small
\begin{tabular}{lrrrrrr}
\toprule
Decision type & Issued & \shortstack{Below\\certificate} & Accepted & \shortstack{Validation\\failure} & Certificate & \shortstack{Hardest-fifth\\failure} \\
\midrule
Ranking & 24 & 24 & 66.54\% & .370\% & 19.76\% & 13.02\% \\
Pruning & 24 & 24 & 67.26\% & .848\% & 22.75\% & 21.47\% \\
\bottomrule
\end{tabular}

\small
\caption{Five-pool protocol at 512 certification blocks ($\delta=0.05$), using the separate-bounds certificate. ``Below certificate'' counts configurations whose validation failure rate stayed below the fixed certificate.}
\label{tab:app-p6}
\end{table}

This protocol uses 128 fitting, 512 certification, 512 validation, 256 cutoff-selection, and 512 hardest-fifth evaluation blocks per configuration. Five acceptance cutoffs and seven thresholds give 35 candidate pairs. The separate-bounds certificate uses 35 upper-tail events for their numerators and five lower-tail events for coverage, hence 40 concentration events. With $n=512$ and $\delta=.05$, the common width is $\epsilon_n=\sqrt{\log(40/\delta)/(2n)}=0.0808$. The certificate is the separate-bounds form $(\widehat U+\epsilon_n)/(\widehat C-\epsilon_n)$, which Appendix~\ref{app:p6-selective} shows to be a looser consequence of the event behind the mechanism certificate whenever its denominator is positive. Ranking certificates range from 11.23 to 58.25\% and pruning certificates from 11.17 to 46.46\%. The last column applies the fixed cutoff to the hardest fifth of margins in a separate pool.

\begin{table}[tbp]
\centering

\small
\begin{tabular}{llrrr}
\toprule
Blocks & Certificate & Value & Acceptance rate & Failure rate \\
\midrule
128 & separate bounds & 40.22\% / 43.96\% & 76.8\% / 73.4\% & 0.99\% / 1.91\% \\
128 & mechanism & 29.08\% / 31.32\% & 70.4\% / 68.4\% & 0.45\% / 0.94\% \\
128 & mechanism, variance-adaptive & 28.60\% / 29.52\% & 67.9\% / 68.4\% & 0.37\% / 0.94\% \\
128 & failure & 16.29\% / 18.62\% & 95.1\% / 88.5\% & 2.18\% / 3.31\% \\
\midrule
256 & separate bounds & 28.58\% / 31.39\% & 70.4\% / 70.1\% & 0.45\% / 1.26\% \\
256 & mechanism & 22.26\% / 24.07\% & 66.2\% / 67.6\% & 0.35\% / 0.86\% \\
256 & mechanism, variance-adaptive & 16.44\% / 17.64\% & 62.0\% / 61.0\% & 0.26\% / 0.75\% \\
256 & failure & 12.18\% / 13.96\% & 93.5\% / 82.7\% & 1.95\% / 2.32\% \\
\midrule
512 & separate bounds & 20.59\% / 22.95\% & 67.0\% / 67.6\% & 0.42\% / 0.86\% \\
512 & mechanism & 16.74\% / 18.65\% & 64.6\% / 63.5\% & 0.33\% / 0.75\% \\
512 & mechanism, variance-adaptive & 9.49\% / 10.91\% & 61.2\% / 56.8\% & 0.26\% / 0.72\% \\
512 & failure & 9.07\% / 10.45\% & 87.8\% / 80.2\% & 1.34\% / 2.01\% \\
\midrule
1{,}024 & separate bounds & 15.13\% / 17.36\% & 63.7\% / 61.0\% & 0.32\% / 0.75\% \\
1{,}024 & mechanism & 12.74\% / 14.65\% & 61.2\% / 61.0\% & 0.26\% / 0.75\% \\
1{,}024 & mechanism, variance-adaptive & 6.13\% / 7.27\% & 57.7\% / 51.1\% & 0.22\% / 0.44\% \\
1{,}024 & failure & 6.75\% / 7.99\% & 81.9\% / 75.2\% & 0.94\% / 1.44\% \\
\bottomrule
\end{tabular}

\small
\caption{Certificates over prespecified sample sizes ($\delta=0.05$). Each entry averages 24 configurations, and each certificate selects its own acceptance rule. Acceptance and failure rates are measured on 1{,}024 independent blocks. Entries are ranking / pruning.}
\label{tab:app-d-certificates}
\end{table}

All 48 configurations satisfy the block-level ordering $0\le F_i\le U_i\le C_i\le1$. Certificates decrease as the number of certification blocks grows, while validation failure rates change little; each certificate may select a different acceptance rule at each sample size. Ranking flips use the same lexicographic tie-breaking as candidate ordering, and pruning treats a non-negative score as pruned.

% ============================================================
\section{Details for E6: When the Gaussian Model Applies}
\label{app:eligibility}

We vary the minimum score correlation over $\{0.05,0.10,0.15\}$, the maximum absolute error--margin correlation over $\{0.4,0.5,0.6\}$, and the maximum Gaussian left-tail multiplier over $\{1.5,2,3\}$, retaining the positive-slope requirement. Across all 27 screens, unrotated GIST one- and two-bit codes pass zero, all four MiniLM variants pass 27, and rotated Cohere one-bit passes 27. The other three Cohere variants pass 18, failing at the 1.5 tail cutoff. Rotated GIST one-bit passes 18 and two-bit passes 27. These measurements locate sensitivity to the heuristic tail screen separately from the assumptions of the analytical inequalities.

\begin{table}[tbp]
\centering

\small
\begin{tabular}{llrrrr}
\toprule
Dataset & Coordinates & \shortstack{Magnitude\\flips} & \shortstack{Threshold-error\\share} & \shortstack{Formula error\\(R / P)} & \shortstack{Coarse\\bound} \\
\midrule
Cohere & original & 1.834\% & 4.66\% & 3.1\% / 4.6\% & 11{,}420 \\
Cohere & rotated & .313\% & 2.09\% & 0.7\% / 2.1\% & 51{,}130 \\
MiniLM & original & .509\% & 1.31\% & 0.3\% / 1.0\% & 13{,}876 \\
MiniLM & rotated & .516\% & 1.37\% & 0.9\% / 1.5\% & 13{,}969 \\
GIST & original & 3.446\% & 100.73\% & 0.0\% / 0.7\% & 19{,}016 \\
GIST & rotated & .267\% & 2.30\% & 2.3\% / 1.4\% & 77{,}414 \\
\bottomrule
\end{tabular}

\small
\caption{Error decomposition and variance formula for two-bit codes. ``Magnitude flips'' counts the fraction of bits that change between per-vector and population thresholds. ``Threshold-error share'' is $\mathrm{Var}(R_{\mathrm{thr}})/\mathrm{Var}(R)$ and may exceed 100\% because components are correlated. ``Formula error'' is the relative error of~\eqref{eq:oracle-variance} against measured ranking / pruning variance. ``Coarse bound'' is the range bound divided by measured idealized variance.}
\label{tab:app-p5}
\end{table}

The MiniLM and GIST source files are unit-normalized, which removes their original radial information. Cohere keeps radial information, with a coefficient of variation of the norm of 4.4\% and 2.03\% of vectors outside a 10\% relative shell.

\begin{table}[tbp]
\centering

\small
\begin{tabular}{lrrrr}
\toprule
Bound & Minimum & Median & Mean & Maximum \\
\midrule
Coarse range & 403 & 22{,}044 & 27{,}109 & 73{,}572 \\
Componentwise Cauchy--Schwarz & 5.91 & 40.84 & 47.59 & 93.13 \\
Replacement (Efron--Stein) & 1.21 & 1.67 & 1.72 & 2.04 \\
\bottomrule
\end{tabular}

\small
\caption{Ratio of three variance bounds to the measured variance over 28 configurations (seven datasets, two coordinate systems, two decision types).}
\label{tab:app-proxies}
\end{table}

All 28 covariances between the bounded-code and bilinear components are negative, with a median cancellation ratio of 0.947. In dimension stress tests from $D=64$ to $2048$, clean, equicorrelated, and block-correlated Gaussian regimes keep the idealized ratio $b_\star/\mathrm{sd}$ between about 0.19 and 0.25. A dense global shock raises it towards 0.88 and is detected at the same time by the Gaussian, shell, and spectrum diagnostics. The rare-spike construction of Theorem~\ref{thm:counterexample} is designed to pass those diagnostics, which makes it the relevant obstruction.

\begin{table}[tbp]
\centering

\small
\begin{tabular}{lrrr}
\toprule
Quantity & Minimum & Median & Maximum \\
\midrule
Scale $b_\star/\mathrm{sd}$, unclipped & .133 & .158 & .324 \\
Sum of increment variances / total variance & 1.047 & 1.086 & 1.102 \\
Validation / calibration variance & .945 & .994 & 1.075 \\
Scale $b_\star/\mathrm{sd}$, clipped at 1\% mass & .130 & .147 & .258 \\
\bottomrule
\end{tabular}

\small
\caption{Parameters of the conditional tail and clipping routes over 24 configurations (six datasets, two codes, two decision types).}
\label{tab:app-cm}
\end{table}

The four disjoint pools contain 900 fitting, 1{,}200 reference, 1{,}200 calibration, and 1{,}200 validation vectors. Nested integration uses 160 first-node groups, 16 second nodes, 12 observed third nodes, and separate reference completions; the clipping cutoffs are the 0.95, 0.975, and 0.99 calibration quantiles. No held-out tail exceeded its bound.

\FloatBarrier
\bibliography{references}

\end{document}